\documentclass[12pt, draftclsnofoot, onecolumn]{IEEEtran}
\makeatletter
\def\endthebibliography{%
  \def\@noitemerr{\@latex@warning{Empty `thebibliography' environment}}%
  \endlist
}
\makeatother
\ifCLASSINFOpdf
   \usepackage[pdftex]{graphicx}
\else
 \usepackage[dvips]{graphicx}
\fi

\usepackage{cite}
\usepackage[cmex10]{amsmath}
\usepackage{mdwmath}
\usepackage{amssymb}
\usepackage{algorithmic}
\usepackage{array}
\usepackage{mdwmath}
\usepackage{mdwtab}
\usepackage{eqparbox}
\usepackage{bm}
\usepackage[tight,footnotesize]{subfigure}
\usepackage{setspace,color}
\usepackage[ruled,vlined]{algorithm2e}
\usepackage[T1]{fontenc}
\usepackage[utf8]{inputenc}
\usepackage{amsthm}
\usepackage{multirow}
\usepackage{amsmath,amssymb}
\usepackage{steinmetz}
\usepackage{mathtools, cuted}
\usepackage{cuted}
\usepackage{flushend}

\newtheorem{theorem}{Theorem}
\newtheorem*{remark}{Remark}

\DeclareMathOperator{\Tr}{Tr}
\begin{document}
\renewenvironment{description}[1][12pt]
  {\list{}{\labelwidth=0pt \leftmargin=#1
   \let\makelabel\descriptionlabel}}
  {\endlist}
 \newcommand\norm[1]{\left\lVert#1\right\rVert}
 \setlength{\abovedisplayskip}{0.1cm}
\setlength{\belowdisplayskip}{0.1cm}
\title{Precoder Design For Multi-group Multicasting with a Common Message}

\author{Ahmet Zahid Yalcin~\IEEEmembership{}and
        Melda Yuksel~\IEEEmembership{}
}

\maketitle
\begin{abstract}

This paper considers precoding for multi-group multicasting with a common message. The multiple antenna base station communicates with $K$ clusters, each with $L$ users. There is a common message destined to all users and a private multicast message for each cluster. We study the weighted sum rate (WSR) maximization problem for two different schemes: (i) the base station transmits the superposition of common and multicast messages, (ii) the base station concatenates the multicast message vector with the common message. We also formulate a second problem, weighted minimum mean square error (WMMSE) minimization, and prove that WSR maximization and WMMSE minimization are equivalent at the optimal solution. Inspired by the WMMSE problem, we suggest a suboptimal algorithm, based on alternating optimization. We apply this algorithm to the two transmission schemes, and understand that there is a fundamental difference between the two. We compare the results with maximal ratio transmission (MRT), and zero-forcing (ZF) precoding, and investigate the effects of the number of base station antennas, the number of groups and the number of users in a group. Finally, we study imperfect successive interference cancellation (SIC) at the receivers and show that the first transmission scheme is more robust.
\end{abstract}



\IEEEpeerreviewmaketitle
\section{Introduction}

With the introduction of the fifth generation (5G) mobile communication systems, mobile communications will diffuse into all areas of daily life. While the fourth generation (4G) is mainly about Internet access, video calls, and cloud computing, 5G will be about smart-health, smart-cities, 4K video, factory automation, self-driving cars, real-time cloud services and more. All these new and different applications have different requirements, such as high data rates, ultra-low latency and ultra-high reliability, or ability to handle massive radio access. 

The multiple-input multiple-output (MIMO) technology is one of the main enablers for high data rates required for enhanced mobile broadband. The capacity of MIMO channels are found in \cite{Foschini1998_OnLimitsOf}, \cite{telatar_99}. These papers reveal that the degrees of freedom in a MIMO channel is limited by the minimum of the number of transmit and receive antennas. Then the degrees of freedom of a communication link between a multiple antenna base station and a single antenna user simply becomes equal to 1. This degrees of freedom limitation can be alleviated by multi-user MIMO (MU-MIMO) transmission \cite{Vaze12}. 

In MU-MIMO systems, multiple users simultaneously receive their own unicast data streams \cite{Goldsmith2006_OnTheOptimality}. The capacity of MU-MIMO channels is achieved by dirty paper coding (DPC) \cite{Weingarten2006_TheCapacityRegion}. However, DPC is a highly complex, non-linear precoding scheme. Therefore, in the literature, simple linear precoding strategies have been investigated for MU-MIMO transmission \cite{Sharif2007_AComparisionTimeSharingDPC}, \cite{Christensen2008_WeightedSumRate}. Although sub-optimal in general, superposition coding (SPC) is also used as a simpler alternative to DPC. It offers good interference management \cite{Zhang2011_AUnifiedTreatmentofSPC}, \cite{Vanka2012_SPCStrategies},  \cite{Joudeh2016_SumRateMax} and harvests some of the benefits of DPC in MU-MIMO systems with applications to non-orthogonal multiple access (NOMA) \cite{ding2014performance, seo2018high}.

While simultaneous multiple unicast data transmission is important, applications such as mobile application updates, mass advertisements and public group communications require multicasting \cite{multicast_1}, \cite{multicast_2}, \cite{multicast_3}. In multicast transmission, the same data is sent to a group of users. Precoding for multicasting is studied in \cite{Jindal2006_CapacityLimitsofMultipleAnt}, \cite{Abdelkader2010_MultipleAntennaMulticasting} and \cite{Chiu2009_TransmitPrecoding}. Precoder design for max-min fairness for multicasting is studied in \cite{Sidiropoulos2006_TransmitBeamforming}. 

Unicast and multicast traffic can also coexist in cellular networks. Scheduling unicast and multicast traffic is considered in \cite{Silva2006_AdaptiveBeamforming} and \cite{Baek2009_AdaptiveTransmission}, while their simultaneous transmission via superposition coding for an enhanced sum rate is studied in \cite{Baek2009_AdaptiveTransmission}. While \cite{Baek2009_AdaptiveTransmission} is for a very limited system with two users, \cite{Yalcin18downlink} studies precoder design problem for a system in which all users have to decode the common message and a subset of these users have to receive their own unicast messages. 

In multi-group multicasting there are multiple groups, where each user in the same group seeks for the same message and different groups receive different messages \cite{Silva2009_LinearTransmitBeamforming}. Precoding for a guaranteed quality of service with minimum transmission power is studied in \cite{Gao2005_GroupOrientedBeamforming} and  \cite{Karidipis2008_QoSMaxMinFairTransmit}. Max-min fair transmit precoding for multi-group multicasting under a total power constraint is also studied in \cite{Karidipis2008_QoSMaxMinFairTransmit}. Weighted fair multi-group multicasting precoding under per antenna power constraints is investigated in \cite{Christopoulos2014_WeightedFair}. A rate splitting approach is suggested for max-min fair transmit precoder design in \cite{Joudeh2017_RateSplittingforMaxMin}. The max-min fair transmit precoding problem for multi-group multicasting in massive MIMO systems is studied in \cite{Zhou2015_JointMulticast} and \cite{Sadeghi2018_MMFMassiveMIMO}. Maximum sum rate for multigroup multicast precoding under per-antenna power constraints is studied in \cite{christopoulos2015multicast}. The authors in \cite{kaliszan2012multigroup} also study sum rate maximization, yet under a total power constraint, and allow for coding over multiple blocks. 

Apart from the above discussion on the objective function (sum rate maximization or max-min fairness) or the system topology (unicast/ multicast transmission etc.), another issue is whether successive interference cancellation (SIC) is perfect or not. In general imperfect channel state information (CSI) and/or hardware impairments result in imperfect SIC. While \cite{Vanka2012_SPCStrategies, ding2014performance, seo2018high, Yalcin18downlink, Joudeh2017_RateSplittingforMaxMin, ekrem2012outer} consider perfect CSI and perfect SIC, \cite{Joudeh2016_SumRateMax} considers imperfect CSI and perfect SIC, and \cite{Kim2015_DesgnOfUserClustering, senel2017optimal} consider imperfect SIC.



All the above multi-group multicasting papers, as well as the others we have encountered in the literature, consider non-overlapping multicasting groups. As an initial step in finding sum rate optimal precoders for the most general problem with overlapping multicast groups with unicast messages, in this work we study multi-group multicasting with common messages. In this system there are non-overlapping multicast groups, yet all groups are also interested in a common message. In this paper we follow the outline listed below:
\begin{enumerate}
\item 
To be able to send the common message along with the private multicast messages, we define two different transmission schemes: (i) the base station transmits the superposition of the common and multicast messages, (ii) the base station regards the common message as another message and concatenates the multicast message vector with the common message.
\item For these transmission schemes, we study the weighted sum rate (WSR) maximization problem. To be able to propose the iterative algorithm described in the next step, we also define the weighted minimum mean square error (WMMSE) problem. We then write the gradient expressions and the Karush-Kuhn-Tucker (KKT) conditions of the Lagrangian equations for both problems. Then, we prove that these two problems are equivalent at the optimum point. Although we cannot find a closed form expression, we state the equations the optimal precoders, receivers and the Lagrange multipliers satisfy. 
\item The WSR maximization problem is non-convex and the optimal solution is hard to find. Inspired by the equivalency between WMMSE minimization and WSR maximization, we propose a low-complexity iterative algorithm for finding a locally optimum solution. The algorithm is based on alternating optimization, which iterates between precoders, mean square error (MSE) weights and receiver structures written for the WMMSE problem. 
We also discuss and display convergence for this algorithm. 
\item We apply this algorithm to both of the transmission schemes, and compare the results with maximal ratio transmission (MRT), and zero-forcing (ZF) precoding. 
\item We also investigate the effect of imperfect SIC on achievable WSR. 
\end{enumerate}
As a result, we learn that there is a fundamental difference between  transmitting the common message via the sum of the private multicast message precoders and via a separate precoder. We find that the two schemes favor common and private multicast messages differently. Also, the former transmission scheme is always better than the latter and is more robust against imperfect SIC. To the best of our knowledge, in the literature, there is no other paper which compares these two schemes to reveal their fundamental differences.

We present the system model, formulate the WSR and WMMSE problems and prove their equivalence in Section \ref{sec:systemmodel}. We propose the iterative algorithm in Section \ref{sec:precoder}. We present the simulation results in Section \ref{sec:simresult}, and conclude the paper in Section \ref{sec:conclusion}.


\section{System Model}\label{sec:systemmodel}

We consider a single cell downlink communication system. The base station is equipped with $M$ transmit antennas. The base station communicates with $K$ clusters and there are $L$ single antenna users in each cluster. Each user belongs to only one cluster. The base station has a common data $s_c$ across all users, and private multicast data ${s}_{u_k}$ ($k=1,\ldots,K$) destined to each of the $K$ clusters. 

In this work, our aim is to understand the conditions under which superposition coding is beneficial. For this purpose, we study two different signal models. In the first model, the common message is superposed onto the private multicast message vector, and in the second model, the common message is appended to the private multicast message vector. 

\begin{description}
\item[Signal Model 1:] 

The base station employs a 2-layer superposition coding scheme, in which the base and enhancement layers respectively carry common and private multicast data. The input data vector is denoted as $\mathbf{s}^{(1)} = {[{s}_1,\ldots,{s}_{K}]}^T$ $\in \mathbb{C}^{K \times 1}$, and each input data stream ${s}_k$, $k= 1,\ldots ,K$ is the sum of common and private multicast data, respectively denoted as ${s}_{c}$ and ${s}_{u_k}$. Thus, the input data vector can be written as $\mathbf{s}^{(1)}=\mathbf{s}_c + \mathbf{s}_u$, where $\mathbf{s}_c = {[{s}_c,\ldots,{s}_c]}^T$ $\in \mathbb{C}^{K \times 1}$ and $\mathbf{s}_u = {[{s}_{u_1},\ldots,{s}_{u_{K}}]}^T$ $\in \mathbb{C}^{K \times 1}$. We assume $s_c$ and all $s_{u_k}$ are independent and
$\mathbb{E}\{{s}_c{s}_c^{\ast}\} = \alpha$ and  $\mathbb{E}\{{s}_{u_k}{s}_{u_k}^{\ast}\} = \bar{\alpha}$. Here, $\bar{\alpha} = 1-\alpha$ and $\alpha$ is the ratio of power allocated to common data. The input data vector $\mathbf{s}^{(1)}$ is linearly processed by a precoder matrix $\mathbf{P}^{(1)} = [\mathbf{p}_1^{(1)},\ldots,\mathbf{p}_{K}^{(1)}]$  $\in \mathbb{C}^{M \times K}$. Each precoding vector $\mathbf{p}_k^{(1)}$ is of size $M \times 1$. 

\item[Signal Model 2:] 

The input data vector is defined as $\mathbf{s}^{(2)} = {[s_c,{s}_{u_1},\ldots,{s}_{u_K}]}^T$ $\in \mathbb{C}^{(K+1) \times 1}$,  where $s_c$ and ${s}_{u_k},k=1,\ldots,K$, are the same as in the above signal model 1. We assume $s_c$ and all $s_{u_k}$ are independent and $\mathbb{E}\{\mathbf{s}^{(2)}{\mathbf{s}^{(2)}}^H\} = \mathbf{I}$. The input data vector $\mathbf{s}^{(2)}$ is linearly processed by a precoder matrix $\mathbf{P}^{(2)} = [\mathbf{p}_c, \mathbf{p}_1^{(2)},\ldots,\mathbf{p}_{K}^{(2)}]$  $\in \mathbb{C}^{M \times (K+1)}$, where both the precoding vector $\mathbf{p}_k^{(2)}$ for each multicast data and $\mathbf{p}_c$ for common data are of size $M \times 1$. 

\end{description}


Then, for signal model $m=1,2$, the overall transmit data vector $\mathbf{x}$ $\in \mathbb{C}^{M \times 1}$ at the base station can be written as
\begin{align}
\mathbf{x}^{(m)} &= \mathbf{P}^{(m)}\mathbf{s}^{(m)} = \mathbf{p}_A^{(m)} {s}_c + \sum_{k=1}^{K} \mathbf{p}_k^{(m)} {s}_{u_k}.\label{precodedSignal}
\end{align}
Here, $\mathbf{p}_A^{(1)} = \sum_{k=1}^{K}\mathbf{p}_k^{(1)}$ and $\mathbf{p}_A^{(2)}  = \mathbf{p}_c$ for signal models 1 and 2, respectively.

As user-level precoding is assumed, there is an average total power constraint,
\begin{align}
\mathbb{E}\{{\mathbf{x}^{(m)}}^{H}{\mathbf{x}^{(m)}}\}&= B^{(m)} \Tr (\mathbf{p}_A \mathbf{p}_A^H) +  C^{(m)} \sum_{k=1}^{K} \Tr(\mathbf{p}_k^{(m)}{\mathbf{p}_k^{(m)}}^H )\leq E_{tx}. \label{pow_const}
\end{align} In (\ref{pow_const}), $(B^{(1)} ,C^{(1)} ) =(\alpha, \bar{\alpha} )$ and $(B^{(2)} ,C^{(2)} ) = (1,1)$ for signal models $1$ and $2$, respectively. Note that, signal model 1 seems to be more restrictive, as $\mathbf{p}_A^{(1)}$ has to be set to the sum of the private multicast data precoders. However, it introduces a new degrees of freedom due to the parameter $\alpha$. The iterative precoder design algorithm we propose in Section~\ref{sec:precoder} designs precoder directions and power levels jointly. The $\alpha$ parameter will serve as an external handle that allows for adjustments in precoder power levels. Because of this difference, in Section~\ref{sec:simresult}, we will observe that the two signal models are fundamentally different.

The received signal at the $l$-th user of the $k$-th cluster can be expressed as
\begin{align}
 {y}_{l,k}^{(m)} &=\mathbf{h}_{l,k} \mathbf{p}_A^{(m)}  {s}_c + \mathbf{h}_{l,k} \sum_{i=1}^{K} \mathbf{p}_i^{(m)}  {s}_{u_i} + {n}_{l,k}.\label{rec_signal}
\end{align} In (\ref{rec_signal}), ${\mathbf{h}_{l,k}}$ $\in \mathbb{C}^{1 \times M}$ is the channel gain vector of the $l$-th user of the $k$-th cluster, $l = 1,...,L$, $k = 1,...,K$. The entries in $\mathbf{h}_{l,k}$ denote complex valued channel gains. 
The noise component ${n}_{l,k}$ is independent and circularly symmetric complex Gaussian random variable with zero mean and unit variance. The base station is assumed to know all $\mathbf{h}_{l,k}$, while the $l$-th user in the $k$-th cluster has to be informed about the composite channel gains $\mathbf{h}_{l,k}\mathbf{p}_A^{(m)}$ and $\mathbf{h}_{l,k}\mathbf{p}_k^{(m)}$. This information can be acquired by utilizing standard training techniques. For example, in time division duplex mode, in the uplink, the users can send training data to the base station, and the base station learns all $\mathbf{h}_{l,k}$. Then, the base station can send training data in the downlink phase \emph{twice}. This way, the users can learn the composite channel gains for the common and the private multicast messages respectively in the first and second downlink transmissions. We further assume that these steps are all error free.



\subsection{Achievable Data Rates}
In this system, all users decode the common message. In addition to this, each user subtracts this common message from its received signal to decode its private multicast message using SIC. Then, the achievable rate for common and private multicast messages for the $l$-th user of the $k$-th cluster are respectively defined as $R_{c,lk}$, and $R_{u,lk}$, $l=1,\ldots,L, k=1,\ldots,K$ and are given as
\begin{align}
R_{c,lk}^{(m)} &= \log \big( 1 + B^{(m)}  {\mathbf{p}_A^{(m)} }^H \mathbf{h}_{l,k}^H {r}_{c,lk}^{(m)^{-1}} \mathbf{h}_{l,k} \mathbf{p}_A^{(m)} \big),\label{rate_cu}\\
R_{u,lk}^{(m)} &= \log \big( {1} + C^{(m)}  {\mathbf{p}_k^{{(m)}}}^H \mathbf{h}_{l,k}^H {r}_{u,lk}^{(m)^{-1}} \mathbf{h}_{l,k} \mathbf{p}_k^{(m)}\big). \label{rate_u}
\end{align} Here ${r}_{c,lk}^{(m)}$ and ${r}_{u,lk}^{(m)}$  are the effective noise variances for common and multicast data at the $l$-th user of the $k$-th cluster for the $m$-th signal model. They can be calculated as
\begin{align}
{r}_{c,lk}^{(m)} &= C^{(m)}  \mathbf{h}_{l,k} \left(\sum_{i=1}^{K}\mathbf{p}_i^{(m)}{\mathbf{p}_i^{(m)}}^H\right) \mathbf{h}_{l,k}^H + 1,\\
{r}_{u,lk}^{(m)} &= \delta^2 B^{(m)}   \mathbf{h}_{l,k} \mathbf{p}_A^{(m)}{\mathbf{p}_A^{(m)} }^H \mathbf{h}_{l,k}^H +  C^{(m)}  \mathbf{h}_{l,k} \left(\sum_{i=1,i\ne k}^{K}\mathbf{p}_i^{(m)}{\mathbf{p}_i^{(m)}}^H\right) \mathbf{h}_{l,k}^H + 1,\label{r_ulk_imperfectSIC}
\end{align}where $\delta$ denotes the amount of residual self interference. If SIC is perfect, $\delta=0$. Note that, overall the achievable rate for common data is determined by the minimum of all $R_{c,lk}^{(m)} $ and the achievable rate for multicast data for group $k$, $s_{u_k}$, is determined by the minimum of all $R_{u,lk}^{(m)}$ in group $k$. Thus, we also define
\begin{align}
R_{c}^{(m)} &= \min_{k=\{1,...,K\}} \min_{l=\{1,...,L\}} R_{c,lk}^{(m)}, \label{eqn:minRc}\\
R_{u,k}^{(m)} &= \min_{l=\{1,...,L\}} R_{u,lk}^{(m)}. \label{eqn:minRu}
\end{align} 
\subsection{The Maximum WSR Problem}
In this paper, our aim is to find the optimal precoders $\mathbf{P}^{(1)}$ and $\mathbf{P}^{(2)}$ for signal models 1 and 2 respectively such that the WSR is maximized subject to a total power constraint. The WSR of the system can be computed as 
\begin{align}
T_1 = \sum_{k=1}^{K}   a_{k} R_{u,k}^{(m)}   + b R_{c}^{(m)}, \label{T_1}
\end{align}where $a_{k}$ and $b$ respectively denote the rate weights that correspond to multicast data at cluster $k$ and common data. The optimization problems $\mathcal{P}_1$ and $\mathcal{P}_2$ are defined as
\begin{align}
&\mathcal{P}_{1}: [{\mathbf{p}_1^{(1)SR}},\ldots,{\mathbf{p}_K^{(1)SR}}] = \arg \max_{\mathbf{p}_k^{(1)}} T_1, \quad \text{s.t. }  (\ref{pow_const}),  \\
&\mathcal{P}_{2}: [\mathbf{p}_c^{SR},\mathbf{p}_1^{(2)SR},\ldots,\mathbf{p}_{K}^{(2)SR}] = \arg \max_{\mathbf{p}_c,\mathbf{p}_k^{(2)}} T_1, \quad \text{s.t. }  (\ref{pow_const}). 
\end{align}We need to convert these two problems to smooth constrained optimization problems because they are non-convex and difficult to solve. Thus, we introduce and constrain two new auxiliary variables $t_k, k = {1,...,K}$ and $z$ as
\begin{align}
t_k  &\leq  a_k R_{u,lk}^{(m)},\; \forall l, \forall k, \label{SR_tk} \\
 z  &\leq  b R_{c,lk}^{(m)},\; \forall l, \forall k. \label{SR_z}
\end{align}Then, we can reformulate $\mathcal{P}_{1}$ and $\mathcal{P}_{2}$ as:
\begin{align}
\mathcal{P}'_{1}:[{\mathbf{p}_1^{(1)SR}},\ldots,{\mathbf{p}_K^{(1)SR}}] = \arg \max_{\mathbf{p}_k^{(1)},t_k,z} T_2, \quad \text{s.t. } (\ref{pow_const}), (\ref{SR_tk}), (\ref{SR_z}), \label{max_SR_2}\\
\mathcal{P}'_{2}:[\mathbf{p}_c^{SR},\mathbf{p}_1^{(2)SR},\ldots,\mathbf{p}_{K}^{(2)SR}] = \arg \max_{\mathbf{p}_c,\mathbf{p}_k^{(2)},t_k,z}  T_2, \quad
\text{s.t. } (\ref{pow_const}), (\ref{SR_tk}), (\ref{SR_z}) \label{max_SR_22},
\end{align}where $T_2 = \sum_{k=1}^{K} t_k +  z$.
\subsection{Error Variance Definitions and the Minimum WMMSE Problem}
The one-to-one correspondence between mutual information and minimum mean square error (MMSE) for Gaussian channels is established in \cite{guo2005mutual}, and for MIMO broadcast channels in \cite{Christensen2008_WeightedSumRate}. To establish such a correspondence in our system setup, we first write the MMSE expressions and then compare them with the achievable rates in (\ref{rate_cu}) and (\ref{rate_u}).

The $l$-th user of the $k$-th cluster first processes its received signal ${y}_{l,k}^{(m)}$, with the common data receiver $W_{l,k}^{(m)}$ to form an estimate of $s_c$, denoted as $\hat{{s}}_{c,lk}= {W}_{l,k}^{(m)} {y}_{l,k}^{(m)}$. In the second stage, the $l$-th user of the $k$-th cluster forms an estimate for the multicast message $u_k$ as $\hat{{s}}_{u_{lk}} = {V}_{l,k}^{(m)}\cdot\left(y_{l,k}^{(m)}-\mathbf{h}_{l,k}\mathbf{p}_{A}^{(m)}s_c+\delta \mathbf{h}_{l,k}\mathbf{p}_{A}^{(m)}s_c\right)$. In all the following analysis, we will assume perfect SIC, $\delta=0$, and in Section~\ref{sec:simresult} we will investigate the effect of nonzero $\delta$.

MSE expressions of common and multicast data for the $l$-th user of the $k$-th cluster are respectively defined as
${\varepsilon}_{c,lk}^{(m)} = \mathbb{E}\left\{\norm{\hat{{s}}_{c,lk} - {s}_c}^2\right\}$, and
${\varepsilon}_{u,lk}^{(m)} = \mathbb{E}\left\{\norm{\hat{{s}}_{u_{lk}} - {s}_{u_k}}^2\right\}$, and for perfect SIC, their closed form expressions can be written as
\begin{align}
{\varepsilon}_{c,lk}^{(m)} &= B^{(m)} \mathbf{h}_{l,k} \mathbf{p}_{A}^{(m)} {\mathbf{p}_{A}^{(m)}}^H \mathbf{h}_{l,k}^H W_{l,k}^{\ast(m)}W_{l,k}^{(m)} + C^{(m)} \sum_{i=1}^{K}\mathbf{h}_{l,k} \mathbf{p}_i^{(m)}{\mathbf{p}_i^{(m)}}^H\mathbf{h}_{l,k}^H W_{l,k}^{\ast(m)}W_{l,k}^{(m)}   \nonumber \\
&\quad \: + W_{l,k}^{\ast(m)}W_{l,k}^{(m)} -  \mathbf{h}_{l,k}\mathbf{p}_{A}^{(m)}W_{l,k}^{(m)} B^{(m)} -B^{(m)}{\mathbf{p}_{A}^{(m)}}^H\mathbf{h}_{l,k}^H W_{l,k}^{\ast(m)} + B^{(m)}, \label{CMSE_u}\\
{\varepsilon}_{u,lk}^{(m)} &= C^{(m)} \sum_{i=1}^{K}\mathbf{h}_{l,k} \mathbf{p}_i^{(m)} {\mathbf{p}_i^{(m)}}^H \mathbf{h}_{l,k}^H V_{l,k}^{\ast(m)} V_{l,k}^{(m)} + V_{l,k}^{\ast(m)} V_{l,k}^{(m)}  -   \mathbf{h}_{l,k}\mathbf{p}_{k}^{(m)}V_{l,k}^{(m)}C^{(m)}\nonumber \\
&\quad \: - C^{(m)} {\mathbf{p}_{k}^{(m)}}^H\mathbf{h}_{l,k}^HV_{l,k}^{\ast(m)}  + C^{(m)}. \label{UMSE_k}
\end{align}The optimal MMSE receivers for common and multicast data are defined as ${W}_{l,k}^{(m)MMSE} = \arg\min_{{W}_{l,k}} {\varepsilon}_{c,lk}^{(m)}$ and ${V}_{l,k}^{(m)MMSE} = \arg\min_{{V}_{l,k}} {\varepsilon}_{u,lk}^{(m)}$.
The closed form expressions for these MMSE receivers are then calculated as
\begin{align}
{W}_{l,k}^{(m)MMSE} &= B^{(m)} \mathbf{p}_A^{(m)^H} \mathbf{h}_{l,k}^H \left( B^{(m)} \mathbf{h}_{l,k} \mathbf{p}_A^{(m)} \mathbf{p}_A^{(m)^H} \mathbf{h}_{l,k}^H + r_{c,lk}^{(m)} \right)^{-1},\label{W_rec}\\
V_{l,k}^{(m)MMSE} &= C^{(m)} \mathbf{p}_k^{(m)^H} \mathbf{h}_{l,k}^H \left( C^{(m)} \mathbf{h}_{l,k} \mathbf{p}_k^{(m)}\mathbf{p}_k^{(m)^H} \mathbf{h}_{l,k}^H + r_{u,lk}^{(m)} \right)^{-1}. \label{V_rec_sic}
\end{align}Given that these MMSE receivers in (\ref{W_rec}) and (\ref{V_rec_sic}) are employed, the resulting error variance expressions in (\ref{CMSE_u}) and (\ref{UMSE_k}) become
\begin{align}
{\varepsilon}_{c,lk}^{(m)MMSE}  &=\left( {\frac{1}{B^{(m)}}} +  \mathbf{p}_A^{(m)^H}\mathbf{h}_{l,k}^H {r}_{c,lk}^{(m)^{-1}} \mathbf{h}_{l,k} \mathbf{p}_A^{(m)} \right)^{-1}, \label{error_cov_common}\\
{\varepsilon}_{u,lk}^{(m)MMSE} &=\left({\frac{1}{C^{(m)}} } + \mathbf{p}_k^{(m)^H}\mathbf{h}_{l,k}^H {r}_{u,lk}^{(m)^{-1}} \mathbf{h}_{l,k} \mathbf{p}_k^{(m)} \right)^{-1}.\label{error_cov_uni}
\end{align}Comparing (\ref{rate_cu}) and (\ref{rate_u}) with (\ref{error_cov_common}) and (\ref{error_cov_uni}) we can write
\begin{align}
R_{c,lk}^{(m)} &= -\log\left(\frac{\varepsilon_{c,lk}^{(m)MMSE}}{ B^{(m)}}\right), \label{rate_cu_errrorCov}\\
R_{u,lk}^{(m)} &= -\log \left( \frac {\varepsilon_{u,lk}^{(m)MMSE}}{C^{(m)}}\right). \label{rate_u_errorCov}
\end{align}
We also define
\begin{align}
{\varepsilon}_{c}^{(m)MMSE}&= \max_{k=\{1,...,K\}} \max_{l=\{1,...,L\}} {\varepsilon}_{c,lk}^{(m)MMSE},\\
{\varepsilon}_{u,k}^{(m)MMSE} &= \max_{l=\{1,...,L\}} {\varepsilon}_{u,lk}^{(m)MMSE}.
\end{align}The WMMSE minimization objective function is then
\begin{align}
Q_1 = \sum_{k=1}^{K}{v}_{k}^{(m)} {\varepsilon}_{u,k}^{(m)MMSE} + {w}^{(m)}{\varepsilon}_{c}^{(m)MMSE}, \label{eqnQ1}
\end{align}where ${w}^{(m)}$ and ${v}_k^{(m)}$ denote the MMSE weights for common data for all users and multicast data at cluster $k$ respectively. The optimization problems $\mathcal{P}_{3}$ and $\mathcal{P}_{4}$ are defined as
\vspace{-0.2cm}
\begin{eqnarray}
&& \mathcal{P}_{3} : [\mathbf{p}_1^{(1)MS},\ldots,\mathbf{p}_{K}^{(1)MS}] =\arg\min_{\mathbf{p}_{k}^{(1)}}  Q_1, \quad \text{s.t. } (\ref{pow_const}), \label{eqnP3}
\end{eqnarray}
\begin{eqnarray}
&& \mathcal{P}_{4} : [\mathbf{p}_c^{MS},\mathbf{p}_1^{(2)MS},\ldots,\mathbf{p}_{K}^{(2)MS}] = \arg\min_{\mathbf{p}_{c},\mathbf{p}_{k}^{(2)}}  Q_1, \quad \text{s.t. } (\ref{pow_const}) \label{eqnP4}.
\end{eqnarray} To solve these non-convex problems, we reformulate both $\mathcal{P}_{3}$ and $\mathcal{P}_{4}$ as smooth constrained optimization problems by introducing auxiliary variables. Without loss of generality, we choose the auxiliary variables as a function of the independent variables $t_k$ and $z$ defined in (\ref{SR_tk}) and (\ref{SR_z}) such that 
\begin{align}
{\varepsilon}_{u,lk}^{(m)MMSE} \leq &e^{-t_k/a_k}, \; \forall l, \forall k, \label{MSE_tk} \\
{\varepsilon}_{c,lk}^{(m)MMSE} \leq &e^{-z/b}, \; \forall l, \forall k. \label{MSE_z}
\end{align} As stated above, this choice does not result in a loss of generality, but is necessary to establish the equivalence between WSR and WMMSE problems. Then, we reformulate $\mathcal{P}_{3}$ and $\mathcal{P}_{4}$ as
\begin{align}
\mathcal{P}'_{3} &: [\mathbf{p}_1^{(1)MS},\ldots,\mathbf{p}_{K}^{(1)MS}] = \arg \min_{\mathbf{p}_{k}^{(1)},t_k,z} Q_2, \quad \text{s.t. }  (\ref{pow_const}), (\ref{MSE_tk}), (\ref{MSE_z}), \label{min_MSE_2}\\
\mathcal{P}'_{4} &: [\mathbf{p}_c^{MS},\mathbf{p}_1^{(2)MS},\ldots,\mathbf{p}_{K}^{(2)MS}]= \arg \min_{\mathbf{p}_{c},\mathbf{p}_{k}^{(2)},t_k,z} Q_2, \quad
\text{s.t. }  (\ref{pow_const}), (\ref{MSE_tk}), (\ref{MSE_z}), \label{min_MSE_22}
\end{align}where
\begin{align}
Q_2 = \sum_{k=1}^{K}{v}_{k}^{(m)} e^{-t_k/ a_k} + w^{(m)} e^{-z/ b}.
\end{align}
In the following, we first prove that the precoders designed for maximum WSR and minimum WMMSE are equivalent at the optimal point. Then we propose an algorithm for precoder design.
\subsection{Gradient Expressions and KKT Conditions for Maximum WSR}
In this section, we study the gradients for the WSR maximization problem. To investigate the stationary points of the problems $\mathcal{P}'_{1}$ and $\mathcal{P}'_{2}$, we formulate the Lagrangian expression as
\begin{align}
f\left(\mathbf{P}^{(m)},t_k,z\right) &= - T_2 + \sum_{k=1}^{K}\sum_{l=1}^{L}\mu_{l,k}^{(m)}(t_k - a_k R_{u,lk}^{(m)}) + \sum_{k=1}^{K}\sum_{l=1}^{L}\eta_{l,k}^{(m)} (z - b R_{c,lk}^{(m)})\nonumber \\
&\;\quad + \lambda^{(m)} \left( B^{(m)}  \norm{\mathbf{p}_A^{(m)}}^2  + C^{(m)}  \sum_{k=1}^{K}\norm{\mathbf{p}_k^{(m)}}^2 - E_{tx}\right) .\label{lagrangian_SR}
\end{align}Here $\lambda^{(m)}$, $\mu_{l,k}^{(m)}$ and $\eta_{l,k}^{(m)}$ are the Lagrange multipliers. We calculate $\nabla_{\mathbf{p}_k^{(m)}} f\left(\mathbf{P}^{(m)},t_k,z\right)$ in Appendix \ref{derive_gradient_f}\footnote{The gradient of a function $f(\mathbf{x})$ with respect to its complex variable $\mathbf{x}$ is denoted as $\nabla_{\mathbf{x}} f(\mathbf{x})$ and its $m$th element is defined as $[\nabla_{\mathbf{x}}f(\mathbf{x})]_{m} = \nabla_{[\mathbf{x}]_{m}}f(\mathbf{x}) = \frac{\partial f(\mathbf{x})}{\partial[\mathbf{x}^\ast]_{m}}$. For detailed derivation rules, we refer the reader to \cite{matrix_cookbook}.} only for signal model $1$ ($m=1$) due to the space limitations. We restate the result.
\vspace{-0.7cm}
\begin{align}
\nabla_{\mathbf{p}_k^{(1)}} f\left(\mathbf{P}^{(1)},t_k,z\right) &= - \sum_{l=1}^{L} \mu_{l,k}^{(1)} a_k \mathbf{h}_{l,k}^H {r}_{u,lk}^{(1)^{-1}} \mathbf{h}_{l,k} \mathbf{p}_k^{(1)} {\varepsilon}_{u,lk}^{(1)MMSE} + \lambda^{(1)}\left(B^{(1)} \mathbf{p}_A^{(1)} + C^{(1)} \mathbf{p}_k^{(1)}\right) \nonumber \\
&\;\quad + \sum_{i=1,i\neq k}^{K} \sum_{l=1}^{L} C^{(1)} \mu_{l,i}^{(1)} a_i \mathbf{h}_{l,i}^H {r}_{u,li}^{(1)^{-1}} \mathbf{h}_{l,i} \mathbf{p}_i^{(1)}{\varepsilon}_{u,li}^{(1)MMSE}\mathbf{p}_i^{(1)^H} \mathbf{h}_{l,i}^H {r}_{u,li}^{(1)^{-1}} \mathbf{h}_{l,i} \mathbf{p}_k^{(1)} \nonumber \\
 &\;\quad + \sum_{i=1}^{K}\sum_{l=1}^{L} C^{(1)}\eta_{l,i}^{(1)} b \mathbf{h}_{l,i}^H {r}_{c,li}^{(1)^{-1}} \mathbf{h}_{l,i} \mathbf{p}_A^{(1)} {\varepsilon}_{c,li}^{(1)MMSE} \mathbf{p}_A^{(1)^H} \mathbf{h}_{l,i}^H {r}_{c,li}^{(1)^{-1}} \mathbf{h}_{l,i}\mathbf{p}_k^{(1)} \nonumber \\
&\;\quad - \sum_{i=1}^{K} \sum_{l=1}^{L}\eta_{l,i}^{(1)} b \mathbf{h}_{l,i}^H {r}_{c,li}^{(1)^{-1}}\mathbf{h}_{l,i}\mathbf{p}_A^{(1)} {\varepsilon}_{c,li}^{(1)MMSE}.  \label{grad_f}
\end{align}
\subsection{Gradient Expressions and KKT Conditions for Minimum WMMSE}
To investigate the stationary points of the minimum WMMSE problems $\mathcal{P}'_3$ and $\mathcal{P}'_4$, we formulate the Lagrangian expression as 
\begin{align}
g(\mathbf{P}^{(m)},t_k,z) &= Q_2 + \sum_{k=1}^{K}\sum_{l=1}^{L} \bar{\mu}_{l,k}^{(m)}v_k^{(m)} (\varepsilon_{u,lk}^{(m)MMSE}- e^{-t_k/a_k}) \nonumber \\
&\;\quad+ \sum_{k=1}^{K}\sum_{l=1}^{L}\bar{\eta}_{l,k}^{(m)} w^{(m)} (\varepsilon_{c,lk}^{(m)MMSE} - e^{-z/b}) \nonumber \\
&\;\quad + \bar{\lambda}^{(m)} \left(B^{(m)} \norm{\mathbf{p}_A^{(m)}}^2  + C^{(m)} \sum_{k=1}^{K} \norm{\mathbf{p}_k^{(m)}}^2 - E_{tx}\right).\label{lagrangian_MSE}
\end{align}Here $\bar{\lambda}^{(m)}$, $\bar{\mu}_{l,k}^{(m)}$ and $\bar{\eta}_{l,k}^{(m)}$ are the Lagrange multipliers. The gradient $\nabla_{\mathbf{p}_k^{(1)} }g(\mathbf{P}^{(1)} ,t_k,z)$ is computed in a similar manner as $\nabla_{\mathbf{p}_k^{(1)} }f(\mathbf{P}^{(1)},t_k,z)$ and is written as in (\ref{grad_g}). 
\begin{eqnarray}
\lefteqn{ \nabla_{\mathbf{p}_k^{(1)}} g\left(\mathbf{P}_k^{(1)},t_k,z\right)}\nonumber\\
&=& - \sum_{l=1}^{L} \bar{\mu}_{l,k}^{(1)}  \mathbf{h}_{l,k}^H {r}_{u,lk}^{(1)^{-1}} \mathbf{h}_{l,k} \mathbf{p}_k^{(1)} {\varepsilon}_{u,lk}^{(1)MMSE} v_k^{(m)} {\varepsilon}_{u,lk}^{(1)MMSE} + \bar{\lambda}^{(1)}\left(B^{(1)}\mathbf{p}_A^{(1)} + C^{(1)}\mathbf{p}_k^{(1)}\right) \nonumber \\
&& - \sum_{i=1}^{K} \sum_{l=1}^{L} \bar{\eta}_{l,i}^{(1)} \mathbf{h}_{l,i}^H {r}_{c,li}^{(1)^{-1}} \mathbf{h}_{l,i} \mathbf{p}_A^{(1)} {\varepsilon}_{c,li}^{(1)MMSE} w^{(m)} {\varepsilon}_{c,li}^{(1)MMSE} \nonumber \\
 &&+ \sum_{i=1,i\neq k}^{K} \sum_{l=1}^{L} C^{(1)}\bar{\mu}_{l,i}^{(1)} \mathbf{h}_{l,i}^H {r}_{u,li}^{(1)^{-1}} \mathbf{h}_{l,i} \mathbf{p}_i^{(1)} {\varepsilon}_{u,li}^{(1)MMSE} v_i^{(m)} {\varepsilon}_{u,li}^{(1)MMSE} \mathbf{p}_i^{(1)^H} \mathbf{h}_{l,i}^H {r}_{u,li}^{(1)^{-1}} \mathbf{h}_{l,i} \mathbf{p}_k^{(1)} \nonumber \\
 &&+ \sum_{i=1}^{K}\sum_{l=1}^{L} C^{(1)} \bar{\eta}_{l,i}^{(1)} \mathbf{h}_{l,i}^H {r}_{c,li}^{(1)^{-1}} \mathbf{h}_{l,i} \mathbf{p}_A^{(1)} {\varepsilon}_{c,li}^{(1)MMSE} w^{(m)} {\varepsilon}_{c,li}^{(1)MMSE} \mathbf{p}_A^{(1)^H} \mathbf{h}_{l,i}^H {r}_{c,li}^{(1)^{-1}}\mathbf{h}_{l,i}\mathbf{p}_k^{(1)}   \label{grad_g}
 \end{eqnarray}
\subsection{Equivalence of WSR and WMMSE Problems}
When (\ref{grad_f}) and (\ref{grad_g}) are compared, it is observed that for a given set of precoders $\mathbf{P}^{(m)}$ and corresponding error variances ${\varepsilon}_{u,lk}^{(m)MMSE}$ and ${\varepsilon}_{c,lk}^{(m)MMSE}$, $\nabla_{\mathbf{p}_k^{(m)}} f$ and  $\nabla_{\mathbf{p}_k^{(m)}} g$ become equal, if the weights $a_k$, $b$, $v_k$ and $w$ are chosen as
\begin{align}
{v}_{k}^{(m)} &= a_{k}{{\varepsilon}_{u,lk}^{(m)MMSE}}^{-1}, \label{MSE_weight1}
\end{align}for all $k$ and $l$ for which $\bar{\mu}_{l,k}^{(m)}>0$, and
\begin{align}
{w}^{(m)} &= b {{\varepsilon}_{c,lk}^{(m)MMSE}}^{-1}, \label{MSE_weight2}
\end{align}for all $k$ and $l$ for which $\bar{\eta}_{l,k}^{(m)}>0$. These relations could also be read as ${{\varepsilon}_{u,lk}^{(m)MMSE}} = a_{k}/v_k^{(m)}$, and $ {{\varepsilon}_{c,lk}^{(m)MMSE}} = b/w^{(m)} $. In other words, at the optimal solution, the MMSE values for the multicast messages are the same within a group, the MMSE values for the common message are the same over the entire set of users, and ${{\varepsilon}_{u,lk}^{(m)MMSE}}$ and $ {{\varepsilon}_{c,lk}^{(m)MMSE}}$ are equal to their own boundaries; i.e. $e^{-t_k/a_k}$, and $e^{-z/b}$ respectively. If this is not possible for any subset of users, then the corresponding Lagrange multipliers are zero. This fact also reveals that $\partial_{t_k} f(\mathbf{P}^{(m)},t_k,z)$~\footnotemark\footnotetext{The notation, $\partial_x f(x)$ denotes $\frac{\partial f(x)}{\partial x}$.} and $\partial_{t_k} g(\mathbf{P}^{(m)},t_k,z)$, and $\partial_z f(\mathbf{P}^{(m)},t_k,z)$ and $\partial_z g(\mathbf{P}^{(m)},t_k,z)$ are equivalent. In conclusion, the two problems have different variable names, but to find the optimal solution, we solve for exactly the same set of equations. Thus, the two optimization problems are equivalent with each other at the optimal solution.

\section{Iterative Precoder Design}\label{sec:precoder}
In this section, we suggest a suboptimal and iterative precoder design algorithm to solve the WSR problem. The WSR problem is non-convex, and hard to solve. It requires efficient algorithms that perform well in practice. Although the WSR problem does not result in an intuitive algorithm, its equivalent WMMSE problem does. As the WMMSE problem is composed of two parts, the transmit precoders and MSE receivers, one can perform alternating optimization between the two. 


To do so, we first define a new optimization problem, same as $\mathcal{P}_3$ (or $\mathcal{P}_4$) defined in (\ref{eqnP3}) (or (\ref{eqnP4})), but instead of defining $Q_1$ in (\ref{eqnQ1}) in terms of $\varepsilon_{c,lk}^{(m)MMSE}$ and $\varepsilon_{u,lk}^{(m)MMSE}$, we use $\varepsilon_{c,lk}^{(m)}$ and $\varepsilon_{u,lk}^{(m)}$ defined in (\ref{CMSE_u}) and (\ref{UMSE_k}). In other words, instead of assuming MMSE receivers, we first allow for any receiver structure.
\vspace{-0.2cm}
\begin{align}
\mathcal{P}_{5} :  \arg &\min_{\mathbf{P}^{(m)},t_k,z,W_{l,k}^{(m)},V_{l,k}^{(m)},v_k^{(m)},w^{(m)}} Q_2\\
\quad  \text{s.t. }  {\varepsilon}_{u,lk}^{(m)} &\leq e^{-t_k/a_k} ,\\
{\varepsilon}_{c,lk}^{(m)} &\leq e^{-z/b} \text{~and~}
(\ref{pow_const}).
\end{align}Then, the new Lagrangian objective function becomes
\begin{align}
h(\mathbf{P}^{(m)},t_k,z) &= Q_2  + \sum_{k=1}^{K}\sum_{l=1}^{L} \xi_{l,k}^{(m)} v_k^{(m)} (\varepsilon_{u,lk}^{(m)} - e^{-t_k/a_k}) + \sum_{k=1}^{K}\sum_{l=1}^{L} \psi_{l,k}^{(m)} w^{(m)} (\varepsilon_{c,lk}^{(m)} - e^{-z/b})  \nonumber\\
&\quad  + \beta^{(m)} \left( B^{(m)} \norm{\mathbf{p}_A^{(m)} }^2  + C^{(m)} \sum_{k=1}^{K} \norm{\mathbf{p}_k^{(m)}}^2 - E_{tx}\right), \label{lagrangian_h}
\end{align}
where $\beta^{(m)}$, $\xi_{l,k}^{(m)}$ and $\psi_{l,k}^{(m)}$ denote the Lagrange multipliers for the $m$-th signal model.

Studying the KKT conditions for this problem, similar to the analysis in the previous section, we can state the following theorem.
\begin{theorem}\label{theorem_Pk}
The common data receiver $W_{l,k}^{(m)}$ in (\ref{W_rec_last}), the multicast data receiver $V_{l,k}^{(m)} $ in (\ref{V_rec_sic_last}) and the Lagrange multiplier $\beta^{(m)}$ in (\ref{lamda}), transmit precoders $\mathbf{p}_k^{(1)}$ in (\ref{Pk_last}), $\mathbf{p}_k^{(2)}$ in (\ref{Pk_last2}), and $\mathbf{p}_c$ in (\ref{Pc_last2}) satisfy the KKT conditions for the optimization problem defined above with the Lagrangian function $h(\mathbf{P}^{(m)},t_k,z)$ defined in (\ref{lagrangian_h}).
\begin{align}
{W}_{l,k}^{(m)} &= B^{(m)} \mathbf{p}_A^{(m)^H} \mathbf{h}_{l,k}^H \left(B^{(m)} \mathbf{h}_{l,k} \mathbf{p}_A^{(m)} \mathbf{p}_A^{(m)^H} \mathbf{h}_{l,k}^H + \sum_{i=1}^{K}C^{(m)} \mathbf{h}_{l,k} \mathbf{p}_i^{(m)}\mathbf{p}_i^{(m)^H} \mathbf{h}_{l,k}^H + 1 \right)^{-1},\label{W_rec_last}\\
{V}_{l,k}^{(m)} &= C^{(m)} \mathbf{p}_k^{(m)^H} \mathbf{h}_{l,k}^H \left(\sum_{i=1}^{K} C^{(m)} \mathbf{h}_{l,k} \mathbf{p}_i^{(m)}\mathbf{p}_i^{(m)^H} \mathbf{h}_{l,k}^H + 1 \right)^{-1},\label{V_rec_sic_last}\\
\beta^{(m)} &= \frac{1}{E_{tx}}\sum_{k=1}^{K} \sum_{l=1}^{L}\left[\xi_{l,k}^{(m)} {v}_k^{(m)} {V}_{l,k}^{(m)} {V}_{l,k}^{(m)^{\ast}} + \psi_{l,k}^{(m)} {w}^{(m)} {W}_{l,k}^{(m)} {W}_{l,k}^{(m)^{\ast}} \right], \label{lamda}\\
\mathbf{p}_k^{(1)} &= \bigg(\beta^{(1)} \mathbf{I} + \sum_{i=1}^{K} \sum_{l=1}^{L}\left( \xi_{l,i}^{(1)} v_i^{(1)} C^{(1)} \mathbf{h}_{l,i}^H {V}_{l,i}^{(1)^{\ast}} {V}_{l,i}^{(1)} \mathbf{h}_{l,i} +  \psi_{l,i}^{(1)} w^{(1)} \mathbf{h}_{l,i}^H {W}_{l,i}^{(1)^{\ast}} {W}_{l,i}^{(1)} \mathbf{h}_{l,i} \right)\bigg)^{-1} \nonumber\\
&\quad \times \bigg[\sum_{i=1}^{K}\sum_{l=1}^{L}\left(\psi_{l,i}^{(1)} w^{(1)} B^{(1)} \mathbf{h}_{l,i}^H {W}_{l,i}^{(1)^{\ast}} - \psi_{l,i}^{(1)} w^{(1)} B^{(1)} \mathbf{h}_{l,i}^H {W}_{l,i}^{(1)^{\ast}} {W}_{l,i}^{(1)} \mathbf{h}_{l,i} (\mathbf{p}_A^{(1)} - \mathbf{p}_k^{(1)})\right) \nonumber \\
& \quad + \sum_{l=1}^{L}\xi_{l,k}^{(1)} v_k^{(1)} C^{(1)} \mathbf{h}_{l,k}^H {V}_{l,k}^{(1)^{\ast}} - \beta^{(1)} B^{(1)} (\mathbf{p}_A^{(1)} - \mathbf{p}_k^{(1)})\bigg],\label{Pk_last}\\
\mathbf{p}_k^{(2)} &= \bigg(\beta^{(2)} \mathbf{I} + \sum_{i=1}^{K} \sum_{l=1}^{L} \xi_{l,i}^{(2)} v_i^{(2)} \mathbf{h}_{l,i}^H {V}_{l,i}^{(2)^{\ast}} {V}_{l,i}^{(2)} \mathbf{h}_{l,i} + \sum_{i=1}^{K} \sum_{l=1}^{L} \psi_{l,i}^{(2)} w^{(2)} \mathbf{h}_{l,i}^H {W}_{l,i}^{(2)^{\ast}} {W}_{l,i}^{(2)} \mathbf{h}_{l,i} \bigg)^{-1} \nonumber \\
& \quad \times \bigg(\sum_{l=1}^{L}\xi_{l,k}^{(2)} v_k^{(2)} \mathbf{h}_{l,k}^H {V}_{l,k}^{(2)^{\ast}} \bigg),\label{Pk_last2}
\end{align}
\begin{align}
\mathbf{p}_c &= \bigg(\beta^{(2)} \mathbf{I} + \sum_{i=1}^{K} \sum_{l=1}^{L} \psi_{l,i}^{(2)} w^{(2)} \mathbf{h}_{l,i}^H {W}_{l,i}^{(2)^{\ast}} {W}_{l,i}^{(2)} \mathbf{h}_{l,i} \bigg)^{-1} \bigg( \sum_{i=1}^{K}\sum_{l=1}^{L}\psi_{l,i}^{(2)} w^{(2)} \mathbf{h}_{l,i}^H {W}_{l,i}^{(2)^{\ast}} \bigg).\label{Pc_last2}
\end{align}
\end{theorem}
\begin{proof}
The proof is provided in Appendix \ref{derive_Pk}.
\end{proof}
\begin{remark}
The receivers $W_{l,k}^{(m)}$ and $V_{l,k}^{(m)}$ in (\ref{W_rec_last}) and (\ref{V_rec_sic_last}) are exactly equal to the MMSE receivers given in (\ref{W_rec}) and (\ref{V_rec_sic}).
\end{remark}
\begin{algorithm}[t]\label{algo_1}
    \SetAlgoLined
    input: $m$, $a_k$, $b$, $\epsilon$, $\alpha$, $E_{tx}$, $\Upsilon$\;
    set $n = 0$, $\left[\mathbf{p}_k^{(m)}\right]^{(n)} = \mathbf{p}_k^{init} \: \forall k $, $\nu = \log_2 (KL)/\epsilon $;\\
    iterate;\\
       \qquad 1. update $n = n + 1$\\
       \qquad 2. compute $\big[{W}_{l,k}^{(m)}\big]^{(n)}$ using (\ref{W_rec_last}) \\
       \qquad 3. compute $\big[{V}_{l,k}^{(m)}\big]^{(n)} \: $ using (\ref{V_rec_sic_last}) \\
       \qquad 4. compute $\varepsilon_{c,lk}^{(m)}$, $\varepsilon_{u,lk}^{(m)}$  using (\ref{CMSE_u}) and (\ref{UMSE_k}) \\
       \qquad 5. compute $w^{(m)}$, $v_k^{(m)}$ using (\ref{MSE_weight1}) and (\ref{MSE_weight2})\\
       \qquad 6. compute $\big[\xi_{l,k}^{(m)}\big]^{(n)}$, $\big[\psi_{l,k}^{(m)}\big]^{(n)}$ using (\ref{xi_l}) and (\ref{psi_l})\\
       \qquad 7. compute $\big[\beta^{(m)}\big]^{(n)}$ using (\ref{lamda}) \\
       \qquad 8. compute $\big[\mathbf{p}_k^{(m)}\big]^{(n)}$ using (\ref{Pk_last}) or (\ref{Pk_last2}), and $\big[\mathbf{p}_c\big]^{(n)}$ using (\ref{Pc_last2})  \\     
       \qquad 9. scale $\big[\mathbf{p}_k^{(m)}\big]^{(n)}$ such that \\
       \qquad \qquad $ B^{(m)} \norm{\big[\mathbf{p}_A^{(m)}\big]^{(n)} }^2 + C^{(m)} \sum_{k=1}^{K_1} \norm{\big[\mathbf{p}_k^{(m)}\big]^{(n)}}^2 = E_{tx} $ \nonumber \\
       \qquad 10. $\mathbf{If}$ $\Tr\big\{\left(\big[\mathbf{P}^{(m)}\big]^{(n)} - \big[\mathbf{P}^{(m)}\big]^{(n-1)}\right) \left(\big[\mathbf{P}^{(m)}\big]^{(n)} -\big[\mathbf{P}^{(m)}\big]^{(n-1)}\right)^H \big\} < \Upsilon $ then\\
       \qquad \qquad terminate \\
       \qquad $\mathbf{else}$\\
       \qquad \qquad go to Step 1
\caption{Iterative WMMSE}
\end{algorithm}
If it was possible to solve the equations stated in Theorem \ref{theorem_Pk} in closed form, then the optimal solution would be obtained. As this is not possible, the iterative WMMSE algorithm, given in Algorithm \ref{algo_1}, iterates between using the receiver structures (\ref{W_rec_last}), (\ref{V_rec_sic_last}), MSE weights (\ref{MSE_weight1}), (\ref{MSE_weight2}) and transmit precoders (\ref{Pk_last}), (\ref{Pk_last2}) and (\ref{Pc_last2}), until convergence to the local optimum; i.e. the error power of two consequtive precoders is small enough. 

Calculating the Lagrange multipliers $\xi_{l,k}^{(m)}$ and $\psi_{l,k}^{(m)}$, $\forall l, k$ and $m$ in Algorithm \ref{algo_1} is not trivial. Applying the exponential penalty method for solving min-max problems defined in \cite{Li2003_ExponentialPenaltyMethod}, in the algorithm in each iteration, we update $\xi_{l,k}^{(m)}$ and $\psi_{l,k}^{(m)}$ according to
\begin{align}
\xi_{l,k}^{(m)} &= \frac{\exp\{\nu (\varepsilon_{u,lk}^{(m)} - e^{-t_k/a_k})\}}{\sum_{l=1}^{L}\exp\{\nu (\varepsilon_{u,lk}^{(m)} - e^{-t_k/a_k})\}},\label{xi_l}\\
\psi_{l,k}^{(m)} &= \frac{\exp\{\nu (\varepsilon_{c,lk}^{(m)} - e^{-z/b})\}}{\sum_{k=1}^{K}\sum_{l=1}^{L}\exp\{\nu(\varepsilon_{c,lk}^{(m)} - e^{-z/b})\}}.\label{psi_l}
\end{align}Here $\nu$ is a constant and as long as $\nu \geq \log(KL)/\epsilon $, the solution is $\epsilon$-optimal. Note that, this choice satisfies the KKT conditions on $\xi_l^{(m)}$ and $\psi_{l,k}^{(m)}$ since $\sum_{l=1}^{L} \xi_{l,k}^{(m)} = 1$, $\sum_{k=1}^{K}\sum_{l=1}^{L} \psi_{l,k}^{(m)} = 1$ and $\xi_{l,k}^{(m)} \geq 0$, $\psi_{l,k}^{(m)} \geq 0$. 

Note that, papers such as \cite{Joudeh2016_SumRateMax} and  \cite{Joudeh2017_RateSplittingforMaxMin} utilize the CVX tool designed for MATLAB \cite{CVX} and do not need to solve for the Lagrange multipliers or for the optimal structures for the transmit precoders, the common data receivers and the multicast data receivers defined in (\ref{W_rec_last})-(\ref{Pc_last2}). When these optimal structures are used, the algorithm finishes in approximately 10 minutes, whereas the algorithm lasts for hours when the CVX tool is employed. 

As there is a total power constraint, and each iteration of the algorithm increases the objective function, the proposed WMMSE algorithm converges to a limit value. Due to the non-convexity of the problem, this limit value is not guaranteed to be the global optimum. However, the algorithm employs the precoders and the MMSE receivers stated in Theorem 1 that satisfy the KKT conditions of the WMMSE problem, and results in a locally optimum solution. Following similar steps as in \cite[Section IV-A]{Christensen2008_WeightedSumRate} and  \cite{Kaleva2016_DecentralizedSumRateMaximization}, one can prove convergence in full detail. 

\section{Simulation Results}\label{sec:simresult}
In this section we provide simulation results to compare the two new precoders Algorithm~\ref{algo_1} generates. We denote these precoders WMMSE1 and WMMSE2 for signal models 1 and 2 respectively. In the following simulations we consider algorithm convergence, and sum rate for different settings. In the simulations, the entries in $\mathbf{h}_{l,k}$ are assumed to be circularly symmetric complex Gaussian distributed random variables with zero mean and unit variance, and are independent and identically distributed. The presented results are averaged over $10^3$ channel realizations. Ideal Gaussian codebooks are used for transmission. In the algorithm, the maximum number of iterations is limited to $100$, and both $\epsilon$ and $\Upsilon$ are set to $10^{-3}$.
\begin{figure}[ht]
\centering
\includegraphics[width=4.6in]{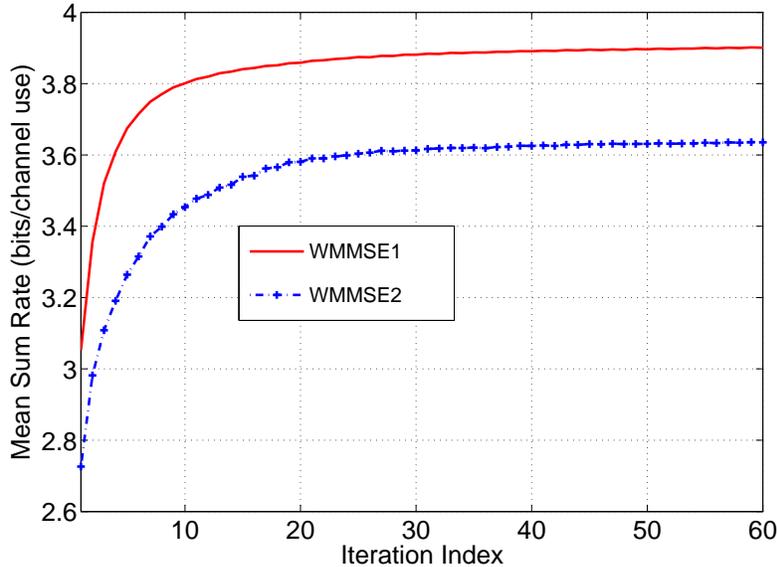}
\caption{Sum-rate convergence performance for $M=K=L=2$. Transmit SNR is set to 15 dB, and $a_k = b = 1$. The optimal $\alpha$ is selected for WMMSE1.}\label{iterSumrate_M2_K2_L2}
\end{figure}
\begin{figure}[ht]
\centering
\includegraphics[width=4.6in]{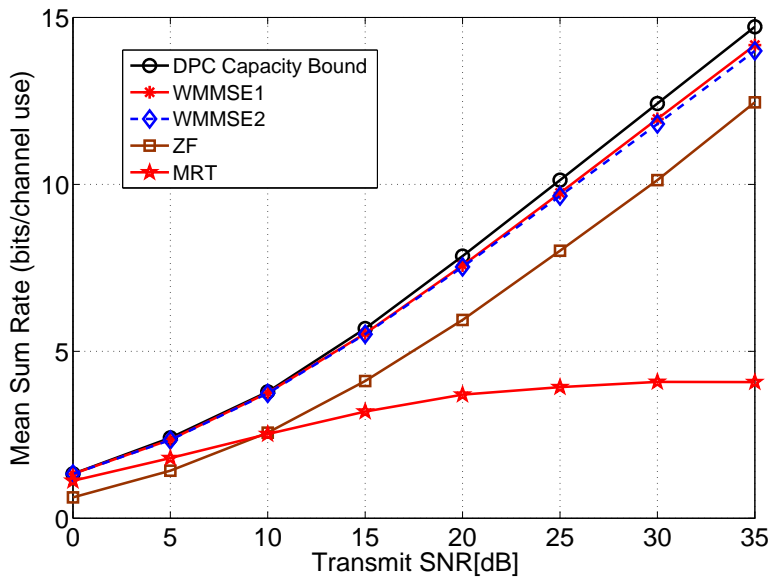}
\caption{The sum-rate $\left(\sum_{k=1}^{K}R_{u,k} + R_{c}\right)$ curves for $M=2$, $K=2$, $L=1$, $a_k=1$, $b=1$. The optimal $\alpha$ is selected for WMMSE1, ZF and MRT.}\label{SR_2_2_1_optAlfa}
\end{figure}

In the simulation results we will provide comparisons with ZF and MRT precoding. Thus, we first describe the two schemes briefly.
\begin{description}
\item[ZF precoding:] The ZF precoder aims to cancel all interference at all users. When, the number of users is less then or equal to the number of transmit antennas ($KL \leq M$), interference cancellation is easily achieved by the ZF precoder, $\mathbf{P}^{ZF}$ $\in \mathbb{C}^{M \times KL}$, given by \cite{Joham2005_LinearTransmitProcessing}
\begin{align}
\mathbf{P}^{ZF} &\triangleq  \sqrt{\frac{E_{tx}}{\Tr((\mathbf{H}\mathbf{H}^H)^{-1})}}\mathbf{H}^H(\mathbf{H}\mathbf{H}^H)^{-1}, \label{ZF_description}
\end{align}where $\mathbf{H} = [\mathbf{h}_{1,1},\mathbf{h}_{1,2}, \ldots,\mathbf{h}_{1,L},\ldots, \mathbf{h}_{K,1},\mathbf{h}_{K,2}, \ldots, \mathbf{h}_{K,L}]^T$ is the composite channel gain matrix of all users, of size $KL \times M$. However, when $KL > M$, total interference cancellation at all users is no longer feasible. Instead, user selection is necessary before designing the ZF precoders. 

In the simulations, we assume $K = M$. Thus, we select a single user in each group (denoted by $l_k$) and design the ZF precoder for this set of users only, where $\mathbf{H}' = [\mathbf{h}_{1,l_1}, \ldots, \mathbf{h}_{K, l_K}]^T$ is the composite channel gain matrix of all selected users. In order to determine the selected set of users, we consider all possible $KL$ channel gain matrices and choose the one, which maximizes the determinant $\left| \mathbf{H}'\mathbf{H}'^H \right|$. This method was shown to minimize MMSE in \cite{Dao2010_UserSelectionAlgo}. We would like to mention that this type of user selection ignores the presence of the common message. The common message is superposed onto the private multicast messages according to signal model 1 during transmission. 

Note that, this method designs ZF precoders according to a single user in each group. Then the achievable multicast rate for a group is determined by the minimum of all multicast rates achievable by each user in a group as in (\ref{eqn:minRu}). Similarly, the achievable common data rate is determined according to (\ref{eqn:minRc}).
\item[MRT precoding:] MRT aims to achieve the highest signal gain at the receivers, and ignores interference. The MRT precoder for cluster $k$ is given in \cite{Sadeghi2018_MMFMassiveMIMO} as
\begin{align}
\mathbf{p}_{k}^{MRT} &= \Gamma \sum_{l=1}^{L}   \mathbf{h}_{l,k},
\end{align} where $\Gamma$ is set to \[ \Gamma = \frac{E_{tx}}{\sum_{k=1}^K \left| \sum_{l=1}^L h_{k,l} \right|^2}\] to satisfy the total power constraint in (\ref{pow_const}). As in the ZF scheme, described above, we assume the common message is superposed onto the private multicast messages.
\end{description}

Fig. \ref{iterSumrate_M2_K2_L2} displays convergence properties for the two iterative WMMSE algorithms. In the figure, total transmit power, $E_{tx}$, is set to $15$ dB and $a_k = b = 1$ and optimal $\alpha$ is chosen for WMMSE1. The initial precoder matrix, $\mathbf{P}^{init}$, is chosen as the zero forcing precoder, as described above. The figure confirms that both algorithms converge fast.  

Note that, the parameter $\alpha$ distributes power unequally among common and multicast data. This results in their average received signal to noise ratios (SNR) to be different. In order to be able to plot the common and multicast data rates on the same graph, in Figs. \ref{SR_2_2_1_optAlfa}-\ref{SR_222_diffWeights} and \ref{SR_441_vs_442_vs_444}-\ref{SR_542_vs_432_vs_423}, the horizontal axis is defined as the total transmit SNR, namely $E_{tx}/\sigma^2$. Here $\sigma^2$ is the noise variance and assumed to be $1$.

Fig. \ref{SR_2_2_1_optAlfa} compares WMMSE1 and WMMSE2 with DPC, ZF and MRT, for $M=K=2$, $L=1$ and $a_k = b=1$. Although this constitutes a limited setting in terms of number of antennas and users, in the literature, capacity region results for MIMO broadcast channels with a common message only exist for this configuration \cite{ekrem2012outer}. For WMMSE1, ZF and MRT, the optimal $\alpha$ is selected for each protocol for every channel realization. It is observed that both WMMSE1 and WMMSE2 achieve the optimal performance at low SNR, and close to optimal at high SNR. In this setting, ZF removes all undesired interference and is parallel to WMMMSE1 and WMMSE2. However, MRT is quite suboptimal as it performs no interference mitigation. 

\begin{figure}[t]
\centering
\includegraphics[width=4.6in]{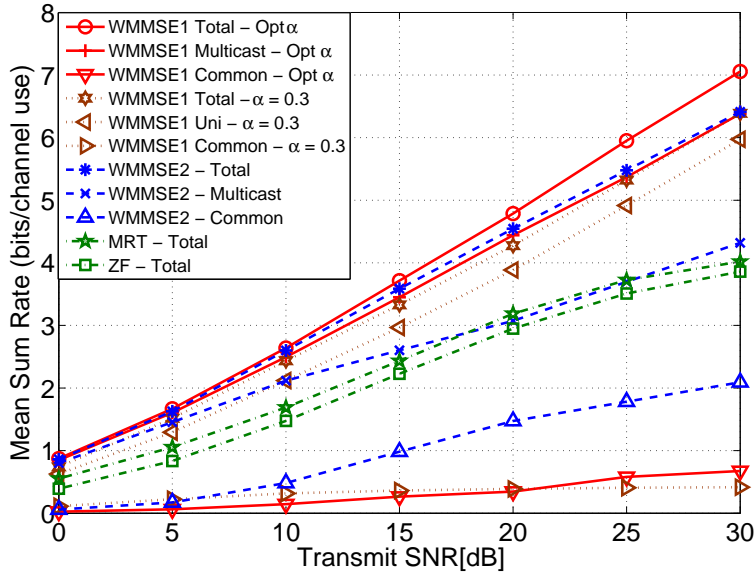}
\caption{The common data rate $\left(R_{c}\right)$, the sum multicast data rate $\left(\sum_{k=1}^{K}R_{u,k}\right)$, and the sum-rate $\left(\sum_{k=1}^{K}R_{u,k} + R_{c}\right)$ curves for $M=K=L=2$, $a_k=b=1$.}\label{SR_222}
\end{figure}
\begin{figure}[ht]
\centering
\includegraphics[width=4.6in]{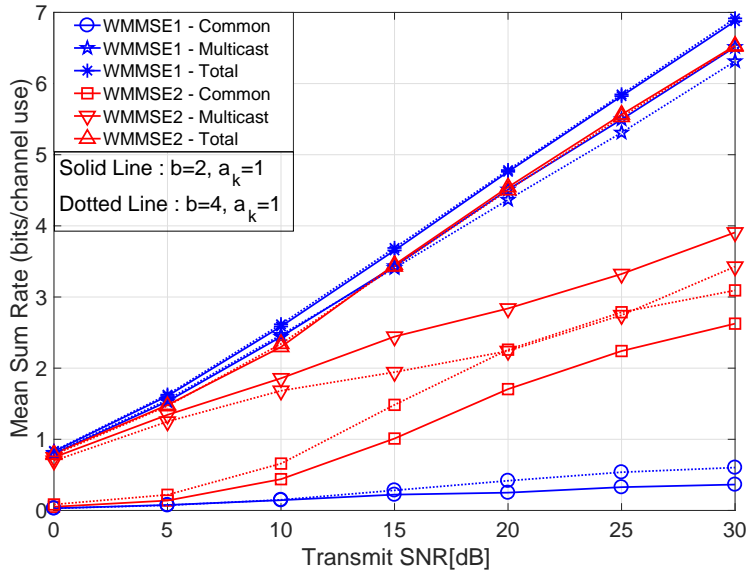}
\caption{The common data rate $\left(R_{c}\right)$, the sum multicast data rate $\left(\sum_{k=1}^{K}R_{u,k}\right)$, and the sum-rate $\left(\sum_{k=1}^{K}R_{u,k} + R_{c}\right)$ curves for $M=K=L=2$ and weights $a_k = 1$ and $b = \{2,4\}$.}\label{SR_222_diffWeights}
\end{figure}
Fig. \ref{SR_222} shows sum rate curves for $M=K=L=2$, and $a_k = b= 1$. The results for WMMSE1 are presented for both optimal $\alpha$ values and for a fixed $\alpha$ value. It is observed that WMMSE1 for optimal $\alpha$ outperforms all protocols in weighted sum rate. The sum multicast rate of WMMSE1, both for optimal and fixed $\alpha$, is larger than that of WMMSE2. On the other hand, WMMSE2 achieves a significantly larger common data rate, at the expense of sum multicast rate. This is because, WMMSE1 mainly designs multicast data precoders and sends the common data via the sum precoder $\mathbf{p}_A^{(1)} = \sum_{k=1}^K \mathbf{p}_k^{(1)}$. However, WMMSE2 assigns a separate precoder to common data $\mathbf{p}_c$. This way, it can achieve higher values for $R_c$, but the loss in multicast data rates is significant and performs worse than WMMSE1 for optimal $\alpha$ in terms of the total weighted sum rate. However, this is not the case if $\alpha$ is fixed. Note that, there is a tradeoff between WMMSE1 and WMMSE2. The precoder definition of WMMSE2 is more general and includes WMMSE1 definition as a special case. In other words, in WMMSE2, the common message precoder can be arbitrary, whereas in WMMSE1 it is contrained to be the sum of the private multicast data precoders. On the other hand, the WMMSE algorithm calculates the best directions and the power levels for all precoders $\mathbf{p}_k$ and $\mathbf{p}_A^{(1)}$ jointly. It does not perform direction and power optimization steps separately. The WMMSE1 algorithm avoids this problem via superposition. WMMSE1 allows the designer to adjust the precoder power levels separately via the parameter $\alpha$. If $\alpha$ is optimized, this new degrees of freedom (power optimization gain) becomes dominant over a more constrained precoder definition and WMMSE1 performs better than WMMSE2. Finally, we observe that MRT performs better than ZF, both MRT and ZF perform poorly in terms of weighted sum rate. When there are multiple users in a group, neither MRT, nor ZF can manage interference well. ZF cancels interference at the selected user in each group, but the interference at the other users in the group are not necessarily cancelled. MRT, on the other hand, behaves as if the group itself is one single user with an equivalent channel gain $\sum_{l=1}^{L} \mathbf{h}_{l,k}$, which does not provide any individual adaptation.

Fig. \ref{SR_222_diffWeights} shows sum rate curves for $M=K=L=2$, $a_k=1$ and $b = \{2,4\}$. The results for WMMSE1 are presented for optimal $\alpha$ values. It is observed that when the weight for the common message, $b$, increases, the common data rate  increases and the sum multicast data rate decreases for both protocols. However, the total data rate does not change.

\begin{figure}[t]
\centering
\includegraphics[width=4.6in]{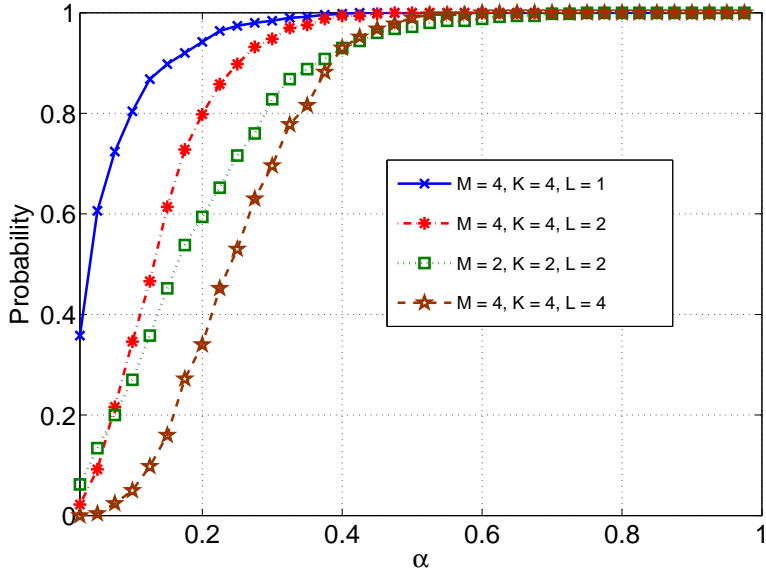}
\caption{Cumulative distribution function of optimal $\alpha$ for different scenarios.}\label{alfaCDF}
\end{figure}
\begin{figure}[t!]
\centering
\includegraphics[width=4.6in]{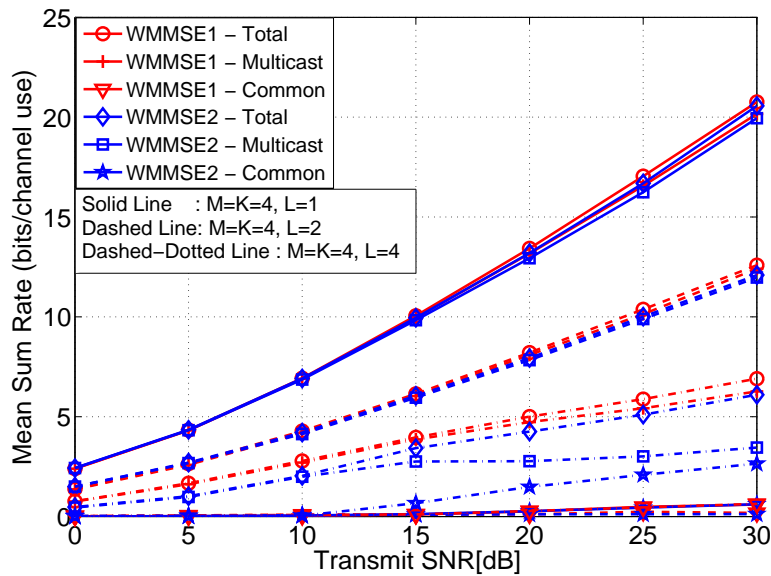}
\caption{The common data rate $\left(R_{c}\right)$, the sum multicast data rate $\left(\sum_{k=1}^{K}R_{u,k}\right)$, and the sum-rate $\left(\sum_{k=1}^{K}R_{u,k} + R_{c}\right)$ curves for $\{M=K=4, L=1\}$, $\{M=K=4, L=2\}$ and $\{M=K=4,L=4\}$ for $a_k=b=1$. 
For WMMSE1 optimal $\alpha$ values are used.}\label{SR_441_vs_442_vs_444}
\end{figure}
\begin{figure}[!h]
\centering
\includegraphics[width=4.6in]{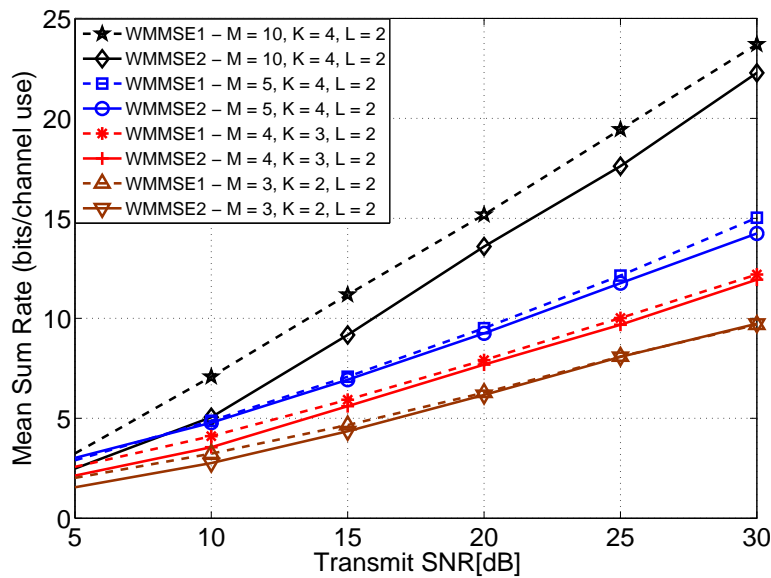}
\caption{The sum-rate $\left(\sum_{k=1}^{K}R_{u,k} + R_{c}\right)$ curves for $\{M=10, K=4, L=2\}$, $\{M=5, K=4, L=2\}$, $\{M=4, K=3, L=2\}$ and $\{M=3, K=2, L=2\}$ for $a_k=b=1$. For WMMSE1 optimal $\alpha$ values are used.}\label{SR_542_vs_432_vs_423}
\end{figure}

Fig. \ref{alfaCDF} shows the cumulative distribution function of optimal $\alpha$ for $M=K=L=2$ and for $M=K=4$ and $L=1,2,4$. The $M=K=L=2$ case is shown to accompany Fig.~\ref{SR_222} and the curves for $M=K=4$ and $L=1,2,4$ accompany Fig.~\ref{SR_441_vs_442_vs_444}. For $M=K=4$ and $L=1,2,4$, as $L$ increases, optimal $\alpha$ values become larger. When $L$ increases, interference management becomes harder and private multicast data rates decrease (as confirmed in Fig.~\ref{SR_441_vs_442_vs_444}). However, in order to sustain the common data rate, larger $\alpha$ values become beneficial. The comparison between $M=K=L=2$ and $M=K=4, L=2$ is also in line with this observation. In the latter scenario, interference management is easier, and thus its cumulative distribution function is to the left $M=K=L=2$.


Fig. \ref{SR_441_vs_442_vs_444} investigates the effect of number of users in a group for $M= K = 4$ and $L=1,2$ or $4$. The figure shows sum rate curves for WMMSE1 for optimal $\alpha$ values. As the number of users in a group increase, both $R_{u,k}$ and $R_{c}$ become the minimum of a larger number of random variables, and become smaller. Thus, the weighted sum rate decreases for increasing number of users. As $L$ increases, interference management at the receivers becomes harder. Moreover, with increasing number of users in a group, the difference between WMMSE1 with optimal $\alpha$ and WMMSE2 becomes more significant. WMMSE1 keeps multicasting rates as high as possible, whereas WMMSE2 prefers increasing the common data rate, which is not sufficient to compensate for the decrease in the sum multicast rate.


\begin{figure}[!t]
\centering
\includegraphics[width=4.6in]{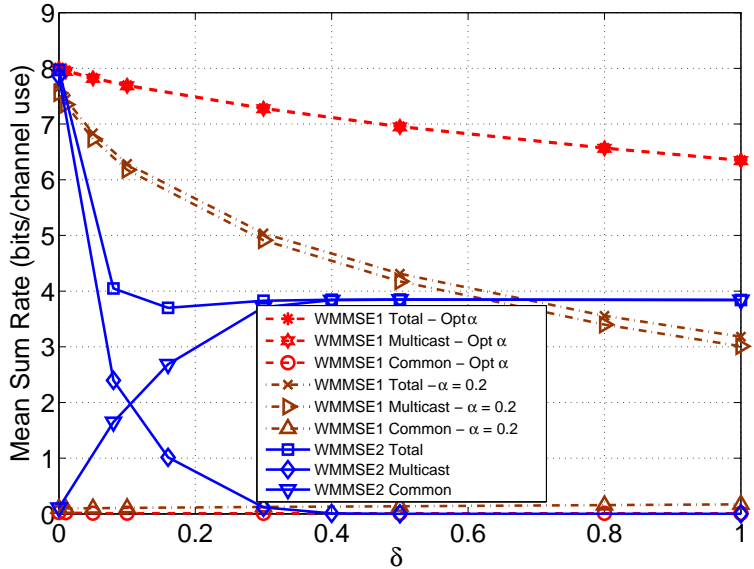}
\caption{The common data rate $\left(R_{c}\right)$, the sum multicast data rate $\left(\sum_{k=1}^{K}R_{u,k}\right)$, and the sum-rate $\left(\sum_{k=1}^{K}R_{u,k} + R_{c}\right)$ curves for $\{M=K=4,L=2\}$ vs. different $\delta$ values for $a_k=b=1$ and transmit SNR is set to 20 dB. For WMMSE1, both optimal $\alpha$ and fixed $\alpha$ is considered.}\label{SR_442_vs_Delta}
\end{figure}

Fig.~\ref{SR_542_vs_432_vs_423} compares WMMSE1 for optimal $\alpha$ and WMMSE2 performance for different number of base station antennas and clusters. In particular, the figure shows  weighted sum rate curves for $\{M=10, K=4, L=2\}$, $\{M=5, K=4, L=2\}$, $\{M=4, K=3, L=2\}$ and $\{M=3, K=2, L=2\}$. We observe that WMMSE1 for optimal $\alpha$ performs better than WMMSE2 in all cases. Moreover, the gains become significant if $M>K \times L$.


Fig. \ref{SR_442_vs_Delta} shows the effect of imperfect SIC on WMMSE1 and WMMSE2 performances. In the figure, $M=K=4$, $L=2$, $a_k=b=1$, transmit SNR is 20 dB and WMMSE1 is displayed for both optimal $\alpha$ and for $\alpha=0.2$. In obtaining the figure, we assume that $\delta$ is unknown at the transmitter and the receivers. Thus the precoders and the receivers designed for perfect SIC are continued to be used. The figure reveals that WMMSE1 (either for optimal or for fixed $\alpha$) is quite robust against imperfect SIC. WMMSE1 with optimal $\alpha$ shows almost no degradation with increasing $\delta$, and WMMSE1 for $\alpha = 0.2$ is better than WMMSE2 for $\delta \leq 0.6$. Note that, residual interference power is related with $\delta^2$, and practically $\delta$ is never as large as 0.6. The figure also emphasizes the fundamental difference between WMMSE1 and WMMSE2 once again. While WMMSE1 has a much higher sum multicast rate than of WMMSE2, WMMSE2 transmits at a higher common data rate than WMMSE1. 



Finally, we find the algorithmic complexities for WMMSE1 and WMMSE2 in Table~\ref{table_1} using the techniques discussed in \cite{Hunger_Complexity2007}. Note that, Table~\ref{table_1} does not show the complexity of searching for the optimal $\alpha$ for WMMSE1. Table~\ref{table_1} reveals that both WMMSE1 and WMMSE2 are cubic in $M$. For $M=K=L=2$, WMMSE1 and WMMSE2 respectively require 1113 and 989 operations. For $M=K=4,L=2$, WMMSE1 and WMMSE2 respectively require 7033 and 5929 operations. 

\begin{table*}[t!]
\caption{Complexity of Proposed Algorithms}\label{table_1}
\begin{tabular}{|c|c|}
  \hline
  \multirow{1}{*}{Algorithm} & Computational Complexity (Number of Matrix Operations)
    \\ \hline
 \multirow{2}{*}{WMMSE1} & \small{$M^3K + M^2(3K^2L + 2K) + M(24K^2L + 21KL +3K -1) + K^2(3L^2 + 24L)$}\\ 
 & \small{$ + K(3L^2 + 38L +3) +1  $} \\ \hline
 \multirow{2}{*}{WMMSE2} & \small{$M^3K + M^2(2K^2L + KL + 2K + 1) + M(20K^2L + 24KL + K) + K^2(3L^2 + 14L)$}\\ 
 & \small{$+ K(3L^2 + 43L) +1$} \\ \hline
  \end{tabular}
\end{table*}

\section{Conclusion}\label{sec:conclusion}

In this paper we study multi-group multicasting with a common message in a downlink MIMO broadcast channel. We assume the multicast groups are disjoint and we investigate the precoder design problem for maximum weighted sum rate. We first prove that weighted sum rate maximization problem is equivalent to the weighted minimum mean square error minimization problem. As both problems are non-convex and highly complex, we suggest a low-complexity, iterative precoder design algorithm inspired by the weighted minimum mean square error minimization problem. The algorithm iterates between receiver design for a given precoder, and then updates the precoder for a given receiver until convergence. We apply this algorithm on two transmission schemes. In the first scheme, the common message and private multicast messages are superposed. In the second scheme, the common message is appended to the private multicast message vector. We show that the first scheme consistently performs better than the second tranmission scheme in all settings. We understand that there is a fundamental difference between the two schemes. Although the second scheme is more general by definition, the first scheme introduces the freedom of power adaptation when employed within the proposed algorithm. Secondly, the first scheme favors multicast transmission, and the second scheme puts more emphasis on the common message. Future work includes studying the effect of successive cancellation order, and investigating the performance of the proposed schemes in a massive MIMO setting with imperfect channel state information at the transmitter.

 \begin{appendices}
 \section{}\label{derive_gradient_f}
 Due to space limitations, in this appendix, we derive $\nabla_{\mathbf{p}_k^{(m)}}f(\mathbf{P}^{(m)},t_k,z)$ for only $m=1$. For notational convenience, we drop the upper indices $(m)$. The Lagrangian objective function is given by
 \begin{align}
 f(\mathbf{P},t_k,z) &= -\sum_{k=1}^{K}t_{k} - z + \underbrace{\lambda \big(\alpha \Tr (\mathbf{p}_A \mathbf{p}_A^H) + \bar{\alpha} \sum_{k=1}^{K} \Tr (\mathbf{p}_k\mathbf{p}_k^H) - E_{tx}\big)}_\text{C} \nonumber \\
 &\:\quad+\underbrace{\sum_{k=1}^{K}\sum_{l=1}^{L}\mu_{l,k}(t_k - a_k R_{u,lk})}_\text{A} + \underbrace{\sum_{k=1}^{K}\sum_{l=1}^{L}\eta_{l,k}(z - b R_{c,lk})}_\text{B}.\label{append_lagrangian_SR}
 \end{align}
 \subsubsection{Gradient of A}To calculate the gradient of A, in (\ref{append_lagrangian_SR}) we need $\nabla_{\mathbf{p}_k} R_{u,lk}$ for both $i=k$ and $i \neq k$ for $l=1,\ldots,L$. Using (\ref{error_cov_uni}) and (\ref{rate_u_errorCov}), we have $\nabla_{\mathbf{p}_k}R_{u,lk}= (\nabla_{\mathbf{p}_k}{\varepsilon}_{u,lk}^{-1}) {\varepsilon}_{u,lk}$. Note that the noise variance ${r}_{u,lk}^{-1}$ is independent from $\mathbf{p}_k$, and $\nabla_{\mathbf{X}}(\mathbf{X}^H\mathbf{A}\mathbf{X})= \mathbf{A}\mathbf{X}$ \cite[ch E.3]{moon_stirling_math}. Then
 \begin{align}
 \nabla_{[\mathbf{p}_k]_{m}}{\varepsilon}_{u,lk}^{-1} &= \mathbf{e}_m^H \mathbf{h}_{l,k}^H {r}_{u,lk}^{-1} \mathbf{h}_{l,k} \mathbf{p}_k, \label{grad_cov_E_eps_k_mn}
 \end{align}where $\mathbf{e}_m$ is the unity column vector with $1$ at the $m^{th}$ element and zeros elsewhere and is size of $M\times 1$. As $[\nabla_{\mathbf{p}_k}R_{u,lk}]_{m} = \nabla_{[\mathbf{p}_k]_{m}}R_{u,lk}= \mathbf{e}_m^H \mathbf{h}_{l,k}^H {r}_{u,lk}^{-1} \mathbf{h}_{l,k} \mathbf{p}_k {\varepsilon}_{u,lk}$, we have
 \begin{align}
 \nabla_{\mathbf{p}_k} R_{u,lk} &= \mathbf{h}_{l,k}^H {r}_{u,lk}^{-1} \mathbf{h}_{l,k} \mathbf{p}_k {\varepsilon}_{u,lk}\label{grad_R_k}.
 \end{align}Next, we compute $\nabla_{\mathbf{p}_k}R_{u,li}$, for $i \neq k$ as
 \begin{align}
 \nabla_{[\mathbf{p}_k]_{m}}R_{u,li}
 &= \mathbf{p}_i^H\mathbf{h}_{l,i}^H \nabla_{[\mathbf{p}_k]_{m}}({r}_{u,li}^{-1}) \mathbf{h}_{l,i} \mathbf{p}_i {\varepsilon}_{u,li}.\label{grad_cov_E_eps_i}
 \end{align}Using $\nabla_{\mathbf{X}} (\mathbf{X}^{-1}) = -\mathbf{X}^{-1}\nabla (\mathbf{X}) \mathbf{X}^{-1}$ \cite{matrix_cookbook}, we can write
 \begin{align}
 \nabla_{[\mathbf{p}_k]_{m}}({r}_{u,li}^{-1})&= -{r}_{u,li}^{-1}\nabla_{[\mathbf{p}_k]_{m}}({r}_{u,li}){r}_{u,li}^{-1}.\label{grad_pkm_invrn2i}
 \end{align}Then we compute
 \begin{align}
 \nabla_{[\mathbf{p}_k]_{m}}({r}_{u,li})&= \mathbf{h}_{l,i} \mathbf{p}_k \bar{\alpha} \mathbf{e}_m^H \mathbf{h}_{l,i}^H.\label{grad_pkm_rn2i}
 \end{align}By combining (\ref{grad_cov_E_eps_i}), (\ref{grad_pkm_invrn2i}) and (\ref{grad_pkm_rn2i}) we have
 \begin{align}
  \nabla_{[\mathbf{p}_k]_{m}}R_{u,li}&= - \bar{\alpha} \mathbf{e}_m^H \mathbf{h}_{l,i}^H {r}_{u,li}^{-1} \mathbf{h}_{l,i} \mathbf{p}_i {\varepsilon}_{u,li}\mathbf{p}_i^H\mathbf{h}_{l,i}^H {r}_{u,li}^{-1} \mathbf{h}_{l,i} \mathbf{p}_k. \label{grad_cov_R_eps_i_mn}
 \end{align}Overall, we have
 \begin{align}
 \nabla_{\mathbf{p}_k}R_{u,li} &= - \bar{\alpha} \mathbf{h}_{l,i}^H {r}_{u,li}^{-1} \mathbf{h}_{l,i} \mathbf{p}_i {\varepsilon}_{u,li}\mathbf{p}_i^H \mathbf{h}_{l,i}^H {r}_{u,li}^{-1} \mathbf{h}_{l,i} \mathbf{p}_k . \label{grad_R_i_last}
 \end{align}Using (\ref{grad_R_k}) and (\ref{grad_R_i_last}), we conclude that the gradient of A can be written as
 \begin{align}
 \nabla_{\mathbf{p}_k} \mathrm{A} &= - \sum_{l=1}^{L} \mu_{l,k} a_k \mathbf{h}_{l,k}^H {r}_{u,lk}^{-1} \mathbf{h}_{l,k} \mathbf{p}_k {\varepsilon}_{u,lk} \nonumber \\
 &\:\quad + \left(\sum_{i=1,i\neq k}^{K} \sum_{l=1}^{L} \bar{\alpha} \mu_{l,k} a_k \mathbf{h}_{l,i}^H {r}_{u,li}^{-1} \mathbf{h}_{l,i} \mathbf{p}_i {\varepsilon}_{u,li}\mathbf{p}_i^H \mathbf{h}_{l,i}^H {r}_{u,li}^{-1} \mathbf{h}_{l,i}\right) \mathbf{p}_k\label{grad_B_last}.
 \end{align}
 \subsubsection{Gradient of B}Secondly we compute the gradient of B defined in (\ref{append_lagrangian_SR}). Note that \\*$\nabla_{\mathbf{p}_k} R_{c,lk} = (\nabla {\varepsilon}_{c,lk}^{-1}){\varepsilon}_{c,lk}$. The gradient $\nabla_{[\mathbf{p}_k]_{m}} {\varepsilon}_{c,lk}^{-1}$ is calculated by applying the chain rule on $\nabla_{[\mathbf{p}_k]_{m}} {\varepsilon}_{c,lk}^{-1}$:
 \begin{align}
 \nabla_{[\mathbf{p}_k]_{m}} {\varepsilon}_{c,lk}^{-1} &= \nabla_{[\mathbf{p}_k]_{m}} \big ({\frac{1}{\alpha}} +  \mathbf{p}_A^H\mathbf{h}_{l,k}^H r_{c,lk}^{-1}\mathbf{h}_{l,k} \mathbf{p}_A\big) \nonumber \\
 &= \mathbf{p}_A^H\mathbf{h}_{l,k}^H \frac{\partial ({r}_{c,lk}^{-1})}{[\partial\mathbf{p}_k^\ast]_{m}}\mathbf{h}_{l,k} \mathbf{p}_A +  \frac{\partial (\mathbf{p}_A^H\mathbf{h}_{l,k}^H)}{[\partial \mathbf{p}_k^\ast]_{m}}{r}_{c,lk}^{-1}\mathbf{h}_{l,k}\mathbf{p}_A + \mathbf{p}_A^H\mathbf{h}_{l,k}^H {r}_{c,lk}^{-1}\frac{\partial (\mathbf{h}_{l,k}\mathbf{p}_A)}{[\partial\mathbf{p}_k^\ast]_{m}} \nonumber  \\
 &= -\mathbf{p}_A^H\mathbf{h}_{l,k}^H {r}_{c,lk}^{-1} \frac{\partial ({r}_{c,lk})}{[\partial\mathbf{p}_k^\ast]_{m}} {r}_{c,lk}^{-1}\mathbf{h}_{l,k}\mathbf{p}_A + \mathbf{e}_m^H\mathbf{h}_{l,k}^H {r}_{c,lk}^{-1}\mathbf{h}_{l,k}\mathbf{p}_A + 0.
 \end{align}The last term is $0$ since $\frac{\partial\mathbf{p}_k}{\partial \mathbf{p}_k^\ast}  = 0$. Now we continue to compute
 \begin{align}
 \nabla_{[\mathbf{p}_k]_m} {r}_{c,lk} &= \frac{\partial \left(\mathbf{h}_{l,k}\left(\sum_{i=1}^{K}\mathbf{p}_i\bar{\alpha} \mathbf{p}_i^H\right)\mathbf{h}_{l,k}^H + 1 \right)}{[\partial \mathbf{p}_k^\ast]_{m}}\nonumber \\
 &= \frac{\partial \left(\mathbf{h}_{l,k} \mathbf{p}_k \bar{\alpha}\mathbf{p}_k^H \mathbf{h}_{l,k}^H + \mathbf{h}_{l,k}\left(\sum_{i=1,i\neq k}^{K}\mathbf{p}_i\bar{\alpha}\mathbf{p}_i^H\right)\mathbf{h}_{l,k}^H + 1 \right)}{[\partial\mathbf{p}_k^\ast]_{m}} \nonumber \\
 &=\mathbf{h}_{l,k}\mathbf{p}_k \bar{\alpha} \mathbf{e}_m^H \mathbf{h}_{l,k}^H,
 \end{align}as $\nabla_{\mathbf{X}}(\mathbf{X}\mathbf{A}\mathbf{X}^H)= \mathbf{X}\mathbf{A}$\cite{moon_stirling_math}. Then, we have
 \begin{align}
 \nabla_{[\mathbf{p}_k]_{m}}{\varepsilon}_{c,lk}^{-1} &= -\mathbf{p}_A^H\mathbf{h}_{l,k}^H {r}_{c,lk}^{-1}\mathbf{h}_{l,k}\mathbf{p}_k \bar{\alpha} \mathbf{e}_m^H \mathbf{h}_{l,k}^H {r}_{c,lk}^{-1}\mathbf{h}_{l,k} \mathbf{p}_A + \mathbf{e}_m^H\mathbf{h}_{l,k}^H {r}_{c,lk}^{-1}\mathbf{h}_{l,k}\mathbf{p}_A, \label{append_invEck}\\
 \nabla_{[\mathbf{p}_k]_{m}}R_{c,lk} &= - \mathbf{e}_m^H \mathbf{h}_{l,k}^H {r}_{c,lk}^{-1}\mathbf{h}_{l,k} \mathbf{p}_A {\varepsilon}_{c,lk} \mathbf{p}_A^H\mathbf{h}_{l,k}^H {r}_{c,lk}^{-1}\mathbf{h}_{l,k}\mathbf{p}_k \bar{\alpha} + \mathbf{e}_m^H\mathbf{h}_{l,k}^H {r}_{c,lk}^{-1}\mathbf{h}_{l,k}\mathbf{p}_A {\varepsilon}_{c,lk}.
  \label{grad_R_eps_cl_mn_part1}
 \end{align}Finally, we obtain
 \begin{align}
 \nabla_{\mathbf{p}_k}R_{c,lk} &= - \mathbf{h}_{l,k}^H {r}_{c,lk}^{-1}\mathbf{h}_{l,k} \mathbf{p}_A {\varepsilon}_{c,lk} \mathbf{p}_A^H\mathbf{h}_{l,k}^H {r}_{c,lk}^{-1}\mathbf{h}_{l,k}\mathbf{p}_k \bar{\alpha} + \mathbf{h}_{l,k}^H {r}_{c,lk}^{-1}\mathbf{h}_{l,k}\mathbf{p}_A {\varepsilon}_{c,lk}, \label{grad_R_eps_cl_last}
 \end{align}and the gradient of B becomes
 \begin{align}
 \nabla_{\mathbf{p}_k} \mathrm{B}
 &= \left (\sum_{i=1}^{K}\sum_{l=1}^{L} \bar{\alpha}\eta_{l,i} b \mathbf{h}_{l,i}^H {r}_{c,li}^{-1}\mathbf{h}_{l,i} \mathbf{p}_A {\varepsilon}_{c,li} \mathbf{p}_A^H\mathbf{h}_{l,i}^H {r}_{c,li}^{-1}\mathbf{h}_{l,i}\right)\mathbf{p}_k \nonumber \\
 &\:\quad - \sum_{i=1}^{K} \sum_{l=1}^{L}\eta_{l,i} b \mathbf{h}_{l,k}^H {r}_{c,lk}^{-1}\mathbf{h}_{l,k}\mathbf{p}_A {\varepsilon}_{c,lk}.
 \label{grad_A_last}
 \end{align}
 \subsubsection{Gradient of C}Thirdly, we compute the gradient of C defined in (\ref{append_lagrangian_SR}) as
 \begin{align}
 \nabla_{\mathbf{p}_k} \mathrm{C}
 &= \lambda\left(\alpha\mathbf{p}_A + \bar{\alpha}\mathbf{p}_k\right),
 \label{grad_C_last}
 \end{align}since $\nabla_{\mathbf{p}_k} \Tr (\mathbf{p}_A \mathbf{p}_A^H) = \mathbf{p}_A$. Finally combining (\ref{grad_B_last}), (\ref{grad_A_last}) and (\ref{grad_C_last}), we have
 \begin{align}
 \nabla_{\mathbf{p}_k} f(\mathbf{P},t_k,z) &= - \sum_{l=1}^{L} \mu_{l,k} a_k \mathbf{h}_{l,k}^H {r}_{u,lk}^{-1} \mathbf{h}_{l,k} \mathbf{p}_k {\varepsilon}_{u,lk} \nonumber \\
 &\:\quad + \left(\sum_{i=1,i\neq k}^{K} \sum_{l=1}^{L} \bar{\alpha} \mu_{l,i} a_i \mathbf{h}_{l,i}^H {r}_{u,li}^{-1} \mathbf{h}_{l,i} \mathbf{p}_i {\varepsilon}_{u,li}\mathbf{p}_i^H \mathbf{h}_{l,i}^H {r}_{u,li}^{-1} \mathbf{h}_{l,i}\right) \mathbf{p}_k \nonumber \\
 &\:\quad + \left (\sum_{i=1}^{K}\sum_{l=1}^{L} \bar{\alpha} \eta_{l,i} b \mathbf{h}_{l,i}^H {r}_{c,li}^{-1}\mathbf{h}_{l,i} \mathbf{p}_A {\varepsilon}_{c,li} \mathbf{p}_A^H\mathbf{h}_{l,i}^H {r}_{c,li}^{-1}\mathbf{h}_{l,i}\right)\mathbf{p}_k \nonumber \\
 &\:\quad - \sum_{i=1}^{K} \sum_{l=1}^{L}\eta_{l,i} b \mathbf{h}_{l,i}^H {r}_{c,li}^{-1}\mathbf{h}_{l,i}\mathbf{p}_A {\varepsilon}_{c,li} + \lambda\left(\alpha\mathbf{p}_A + \bar{\alpha}\mathbf{p}_k\right)
 \label{grad_f_append}.
 \end{align}
  \section{}\label{derive_Pk}
   In this appendix, we prove Theorem \ref{theorem_Pk}. In subsections A and B, we do the derivations respectively for signal models 1 and 2. For notational convenience, we drop the upper index $(m)$.
  \subsection{Derivations for Signal Model $1$ $(m=1)$}
  Taking the derivative of the objective function $h$ in (\ref{lagrangian_h}) with respect to ${W}_{l,k}$, then equating it to zero, we obtain the following equation
  \begin{align}
  \psi_{l,k} w \alpha \mathbf{p}_A^H \mathbf{h}_{l,k}^H  &= \sum_{i=1}^{K}\psi_{l,k} w \bar{\alpha} \mathbf{h}_{l,k} \mathbf{p}_i\mathbf{p}_i^H \mathbf{h}_{l,k}^H  {W}_{l,k}  +\psi_{l,k} w \alpha \mathbf{h}_{l,k} \mathbf{p}_A \mathbf{p}_A^H \mathbf{h}_{l,k}^H {W}_{l,k} + \psi_{l,k} {w} {W}_{l,k}. \label{W_power}
  \end{align}This leads to
  \begin{align}
  {W}_{l,k} &= \alpha\mathbf{p}_A^H \mathbf{h}_{l,k}^H \big(\sum_{i=1}^{K} \bar{\alpha} \mathbf{h}_{l,k} \mathbf{p}_i\mathbf{p}_i^H \mathbf{h}_{l,k}^H + \alpha \mathbf{h}_{l,k} \mathbf{p}_A \mathbf{p}_A^H \mathbf{h}_{l,k}^H + 1 \big)^{-1}\label{W_rec_append}.
  \end{align}Similarly, we take the derivative of (\ref{lagrangian_h}) with respect to ${V}_{l,k}$ and equating it to zero, we have
  \begin{align}
  \xi_{l,k} v_k \bar{\alpha} \mathbf{p}_k^H \mathbf{h}_{l,k}^H  &= \sum_{i=1}^{K} \xi_{l,k} {v}_k \mathbf{h}_{l,k} \mathbf{p}_i\bar{\alpha}\mathbf{p}_i^H \mathbf{h}_{l,k}^H {V}_{l,k} + \xi_{l,k} {v}_k {V}_{l,k}, \label{V_power_sic}
  \end{align}and this leads to
  \begin{align}
  {V}_{l,k} &= \bar{\alpha} \mathbf{p}_k^H \mathbf{h}_{l,k}^H\big(\sum_{i=1}^{K} \bar{\alpha} \mathbf{h}_{l,k} \mathbf{p}_i\mathbf{p}_i^H \mathbf{h}_{l,k}^H + 1 \big)^{-1}.\label{V_rec_sic_append}
  \end{align}Then, taking the gradient of (\ref{lagrangian_h}) with respect to ${\mathbf{p}_k}$, and equating it to zero, we have the following equation
  \begin{eqnarray}
  \lefteqn{\sum_{l=1}^{L}\xi_{l,k} v_k \bar{\alpha} \mathbf{h}_{l,k}^H {V}_{l,k}^{\ast}  + \sum_{i=1}^{K}\sum_{l=1}^{L}\psi_{l,i} w \alpha \mathbf{h}_{l,i}^H {W}_{l,i}^{\ast} }\nonumber \\
  &=& \sum_{i=1}^{K} \sum_{l=1}^{L} \xi_{l,i} v_i \bar{\alpha} \mathbf{h}_{l,i}^H {V}_{l,i}^{\ast} {V}_{l,i} \mathbf{h}_{l,i} \mathbf{p}_k + \sum_{i=1}^{K} \sum_{l=1}^{L} \psi_{l,i} w \bar{\alpha} \mathbf{h}_{l,i}^H {W}_{l,i}^{\ast} {W}_{l,i} \mathbf{h}_{l,i} \mathbf{p}_k  \nonumber \\
  &&+ \sum_{i=1}^{K} \sum_{l=1}^{L} \psi_{l,i} w \alpha \mathbf{h}_{l,i}^H {W}_{l,i}^{\ast} {W}_{l,i} \mathbf{h}_{l,i} \mathbf{p}_A + \beta\big( \alpha\mathbf{p}_A + \bar{\alpha}\mathbf{p}_k\big).\label{Pk_power}
  \end{eqnarray}Then,
  \begin{align}
  \mathbf{p}_k &= \big(\beta \mathbf{I} + \sum_{i=1}^{K} \sum_{l=1}^{L} \xi_{l,i} v_i \bar{\alpha} \mathbf{h}_{l,i}^H {V}_{l,i}^{\ast} {V}_{l,i} \mathbf{h}_{l,i} + \sum_{i=1}^{K} \sum_{l=1}^{L} \psi_{l,i} w \mathbf{h}_{l,i}^H {W}_{l,i}^{\ast} {W}_{l,i} \mathbf{h}_{l,i} \big)^{-1} \nonumber\\
  &\quad \times \big[\sum_{l=1}^{L}\xi_{l,k} v_k \bar{\alpha} \mathbf{h}_{l,k}^H {V}_{l,k}^{\ast}  + \sum_{i=1}^{K}\sum_{l=1}^{L}\psi_{l,i} w \alpha \mathbf{h}_{l,i}^H {W}_{l,i}^{\ast} - \beta\alpha(\mathbf{p}_A-\mathbf{p}_k) \nonumber\\
  &\quad - \sum_{i=1}^{K} \sum_{l=1}^{L} \psi_{l,i} w \alpha \mathbf{h}_{l,i}^H {W}_{l,i}^{\ast} {W}_{l,i} \mathbf{h}_{l,i} (\mathbf{p}_A-\mathbf{p}_k) \big].\label{Pk_append}
  \end{align}
  To calculate $\beta$, we post-multiply both sides of (\ref{W_power}) by ${W}_{l,k}^{\ast}$ and (\ref{V_power_sic}) by ${V}_{l,k}^{\ast}$, and perform $\sum_{k=1}^{K} \sum_{l=1}^{L}$ on (\ref{W_power}) and (\ref{V_power_sic}) on both sides. Then we obtain
  \begin{align}
  \sum_{k=1}^{K} \sum_{l=1}^{L}\psi_{l,k} w \alpha \mathbf{p}_A^H \mathbf{h}_{l,k}^H{W}_{l,k}^{\ast}  &= \sum_{k=1}^{K} \sum_{l=1}^{L} \sum_{i=1}^{K}\psi_{l,k} w \bar{\alpha} \mathbf{h}_{l,k} \mathbf{p}_i\mathbf{p}_i^H \mathbf{h}_{l,k}^H  {W}_{l,k}{W}_{l,k}^{\ast}  \nonumber \\
  &\quad+ \sum_{k=1}^{K} \sum_{l=1}^{L}\psi_{l,k} w \alpha \mathbf{h}_{l,k} \mathbf{p}_A \mathbf{p}_A^H \mathbf{h}_{l,k}^H {W}_{l,k} {W}_{l,k}^{\ast} + \sum_{k=1}^{K} \sum_{l=1}^{L}\psi_{l,k} {w} {W}_{l,k} {W}_{l,k}^{\ast}\label{W_sum_pow}\\
  \sum_{k=1}^{K} \sum_{l=1}^{L} \xi_{l,k} v_k \bar{\alpha} \mathbf{p}_k^H \mathbf{h}_{l,k}^H {V}_{l,k}^{\ast} &= \sum_{k=1}^{K} \sum_{l=1}^{L}\sum_{i=1}^{K} \xi_{l,k} {v}_k \mathbf{h}_{l,k} \mathbf{p}_i\bar{\alpha}\mathbf{p}_i^H \mathbf{h}_{l,k}^H {V}_{l,k} {V}_{l,k}^{\ast} + \sum_{k=1}^{K} \sum_{l=1}^{L}\xi_{l,k} {v}_k {V}_{l,k} {V}_{l,k}^{\ast}. \label{V_sum_pow}
  \end{align}The summation of (\ref{W_sum_pow}) and (\ref{V_sum_pow}) leads to
  \begin{eqnarray}
  \lefteqn{\sum_{k=1}^{K} \sum_{l=1}^{L}\psi_{l,k} w \alpha \mathbf{p}_A^H \mathbf{h}_{l,k}^H{W}_{l,k}^{\ast} + \sum_{k=1}^{K} \sum_{l=1}^{L} \xi_{l,k} v_k \bar{\alpha} \mathbf{p}_k^H \mathbf{h}_{l,k}^H {V}_{l,k}^{\ast}}\nonumber \\
  &=&\sum_{k=1}^{K} \sum_{l=1}^{L} \sum_{i=1}^{K}\psi_{l,k} w \bar{\alpha} \mathbf{h}_{l,k} \mathbf{p}_i\mathbf{p}_i^H \mathbf{h}_{l,k}^H  {W}_{l,k}{W}_{l,k}^{\ast} + \sum_{k=1}^{K} \sum_{l=1}^{L}\psi_{l,k} w \alpha \mathbf{h}_{l,k} \mathbf{p}_A \mathbf{p}_A^H \mathbf{h}_{l,k}^H {W}_{l,k} {W}_{l,k}^{\ast} \nonumber \\
  && + \sum_{k=1}^{K} \sum_{l=1}^{L}\psi_{l,k} {w} {W}_{l,k} {W}_{l,k}^{\ast} + \sum_{k=1}^{K} \sum_{l=1}^{L}\sum_{i=1}^{K} \xi_{l,k} {v}_k \mathbf{h}_{l,k} \mathbf{p}_i\bar{\alpha}\mathbf{p}_i^H \mathbf{h}_{l,k}^H {V}_{l,k} {V}_{l,k}^{\ast} \nonumber \\
  && + \sum_{k=1}^{K} \sum_{l=1}^{L}\xi_{l,k} {v}_k {V}_{l,k} {V}_{l,k}^{\ast}.\label{VW_last_power}
  \end{eqnarray}On the other hand, pre-multiplying (\ref{Pk_power}) with $\mathbf{p}_k^H$ and summing over $k$ from $1$ to $K$, we have
  \begin{eqnarray}
  \lefteqn{\sum_{k=1}^{K}\sum_{l=1}^{L}\xi_{l,k} v_k \bar{\alpha} \mathbf{p}_k^H\mathbf{h}_{l,k}^H {V}_{l,k}^{\ast}  + \sum_{k=1}^{K} \sum_{i=1}^{K}\sum_{l=1}^{L}\psi_{l,i} w \alpha \mathbf{p}_k^H \mathbf{h}_{l,i}^H {W}_{l,i}^{\ast} }\nonumber \\
  &=& \sum_{k=1}^{K} \sum_{i=1}^{K} \sum_{l=1}^{L} \xi_{l,i} v_i \bar{\alpha} \mathbf{p}_k^H \mathbf{h}_{l,i}^H {V}_{l,i}^{\ast} {V}_{l,i} \mathbf{h}_{l,i} \mathbf{p}_k + \sum_{k=1}^{K} \sum_{i=1}^{K} \sum_{l=1}^{L} \psi_{l,i} w \bar{\alpha} \mathbf{p}_k^H \mathbf{h}_{l,i}^H {W}_{l,i}^{\ast} {W}_{l,i} \mathbf{h}_{l,i} \mathbf{p}_k  \nonumber \\
  &&+ \sum_{k=1}^{K} \sum_{i=1}^{K} \sum_{l=1}^{L} \psi_{l,i} w \alpha \mathbf{p}_k^H \mathbf{h}_{l,i}^H {W}_{l,i}^{\ast} {W}_{l,i} \mathbf{h}_{l,i} \mathbf{p}_A + \sum_{k=1}^{K}\mathbf{p}_k^H\beta\big( \alpha\mathbf{p}_A + \bar{\alpha} \mathbf{p}_k\big).\label{Pk_last_power}
  \end{eqnarray}Comparing the left sides of (\ref{VW_last_power}) and (\ref{Pk_last_power}), we observe that they are equal. Then, the right sides are also equal to each other. As we assume that the power constraint in (\ref{pow_const}) is satisfied with equality, (\ref{VW_last_power}) = (\ref{Pk_last_power}) leads to
  \begin{align}
  \beta &= \frac{1}{E_{tx}}\left[\sum_{k=1}^{K} \sum_{l=1}^{L}\xi_{l,k} {v}_k {V}_{l,k} {V}_{l,k}^{\ast} + \sum_{k=1}^{K} \sum_{l=1}^{L}\psi_{l,k} {w} {W}_{l,k} {W}_{l,k}^{\ast} \right]. \label{lamda_append}
  \end{align}
  \subsection{Derivations for Signal Model $2$ $(m=2)$}
First, we would like to remind that in the second signal model, $\mathbf{p}_A=\mathbf{p}_c$. Taking the derivative of the objective function $h$ in (\ref{lagrangian_h}) with respect to ${W}_{l,k}$, then equating it to zero, we obtain the following equation 
  \begin{align}
  \psi_{l,k} w \mathbf{p}_c^H \mathbf{h}_{l,k}^H  &= \sum_{i=1}^{K}\psi_{l,k} w \mathbf{h}_{l,k} \mathbf{p}_i\mathbf{p}_i^H \mathbf{h}_{l,k}^H  {W}_{l,k}  +\psi_{l,k} w \mathbf{h}_{l,k} \mathbf{p}_c \mathbf{p}_c^H \mathbf{h}_{l,k}^H {W}_{l,k} + \psi_{l,k} {w} {W}_{l,k}. \label{W_power2}
  \end{align}This leads to
  \begin{align}
  {W}_{l,k} &= \mathbf{p}_c^H \mathbf{h}_{l,k}^H \big(\sum_{i=1}^{K}  \mathbf{h}_{l,k} \mathbf{p}_i\mathbf{p}_i^H \mathbf{h}_{l,k}^H +\mathbf{h}_{l,k} \mathbf{p}_c \mathbf{p}_c^H \mathbf{h}_{l,k}^H + 1 \big)^{-1}\label{W_rec_append2}.
  \end{align}Similarly, we take the derivative of (\ref{lagrangian_h}) with respect to ${V}_{l,k}$ and equating it to zero, we have
  \begin{align}
  \xi_{l,k} v_k \mathbf{p}_k^H \mathbf{h}_{l,k}^H  &= \sum_{i=1}^{K} \xi_{l,k} {v}_k \mathbf{h}_{l,k} \mathbf{p}_i\mathbf{p}_i^H \mathbf{h}_{l,k}^H {V}_{l,k} + \xi_{l,k} {v}_k {V}_{l,k}, \label{V_power_sic2}
  \end{align}and this leads to
  \begin{align}
  {V}_{l,k} &= \mathbf{p}_k^H \mathbf{h}_{l,k}^H\big(\sum_{i=1}^{K}  \mathbf{h}_{l,k} \mathbf{p}_i\mathbf{p}_i^H \mathbf{h}_{l,k}^H + 1 \big)^{-1}.\label{V_rec_sic_append2}
  \end{align}Taking the gradient of (\ref{lagrangian_h}) with respect to ${\mathbf{p}_k}$ and ${\mathbf{p}_c}$, and equating it to zero, we have the following equation
  \begin{eqnarray}
  \sum_{l=1}^{L}\xi_{l,k} v_k  \mathbf{h}_{l,k}^H {V}_{l,k}^{\ast} =\sum_{i=1}^{K} \sum_{l=1}^{L} \xi_{l,i} v_i  \mathbf{h}_{l,i}^H {V}_{l,i}^{\ast} {V}_{l,i} \mathbf{h}_{l,i} \mathbf{p}_k + \sum_{i=1}^{K} \sum_{l=1}^{L} \psi_{l,i} w  \mathbf{h}_{l,i}^H {W}_{l,i}^{\ast} {W}_{l,i} \mathbf{h}_{l,i} \mathbf{p}_k  + \beta \mathbf{p}_k.\label{Pk_power2}
 \end{eqnarray}
 \begin{eqnarray}
  \sum_{k=1}^{K} \sum_{l=1}^{L} \psi_{l,k} w  \mathbf{h}_{l,k}^H {W}_{l,k}^{\ast} =  \sum_{k=1}^{K} \sum_{l=1}^{L} \psi_{l,k} w  \mathbf{h}_{l,k}^H {W}_{l,k}^{\ast} {W}_{l,k} \mathbf{h}_{l,k} \mathbf{p}_c  + \beta \mathbf{p}_c.\label{Pc_power2}
  \end{eqnarray}
  Then,
  \begin{align}
  \mathbf{p}_k &= \left(\beta \mathbf{I} + \sum_{i=1}^{K} \sum_{l=1}^{L} \xi_{l,i} v_i \mathbf{h}_{l,i}^H {V}_{l,i}^{{\ast}} {V}_{l,i} \mathbf{h}_{l,i} + \sum_{i=1}^{K} \sum_{l=1}^{L} \psi_{l,i} w \mathbf{h}_{l,i}^H {W}_{l,i}^{{\ast}} {W}_{l,i} \mathbf{h}_{l,i} \right)^{-1} 
 \left(\sum_{l=1}^{L}\xi_{l,k} v_k \mathbf{h}_{l,k}^H {V}_{l,k}^{{\ast}} \right),\label{Pk2_append}\\
\mathbf{p}_c &= \left(\beta \mathbf{I} + \sum_{i=1}^{K} \sum_{l=1}^{L} \psi_{l,i} w \mathbf{h}_{l,i}^H {W}_{l,i}^{{\ast}} {W}_{l,i} \mathbf{h}_{l,i} \right)^{-1} \left( \sum_{i=1}^{K}\sum_{l=1}^{L}\psi_{l,i} w \mathbf{h}_{l,i}^H {W}_{l,i}^{{\ast}} \right).\label{Pc_append}
  \end{align}
  To calculate $\beta$ for signal model $2$, we followed the similar steps in the previous subsection and obtained the following expression,
  \begin{align}
  \beta &= \frac{1}{E_{tx}}\left[\sum_{k=1}^{K} \sum_{l=1}^{L}\xi_{l,k} {v}_k {V}_{l,k} {V}_{l,k}^{\ast} + \sum_{k=1}^{K} \sum_{l=1}^{L}\psi_{l,k} {w} {W}_{l,k} {W}_{l,k}^{\ast} \right]. \label{lamda_append2}
  \end{align} 
 
 \end{appendices}

\section*{Acknowledgment}
The authors would like to thank the anonymous reviewers for their valuable comments that improved the results.

\ifCLASSOPTIONcaptionsoff
  \newpage
\fi



%

\bibliographystyle{IEEEtran} 
\bibliography{referencesMult}

\begin{thebibliography}{10}
\providecommand{\url}[1]{#1}
\csname url@samestyle\endcsname
\providecommand{\newblock}{\relax}
\providecommand{\bibinfo}[2]{#2}
\providecommand{\BIBentrySTDinterwordspacing}{\spaceskip=0pt\relax}
\providecommand{\BIBentryALTinterwordstretchfactor}{4}
\providecommand{\BIBentryALTinterwordspacing}{\spaceskip=\fontdimen2\font plus
\BIBentryALTinterwordstretchfactor\fontdimen3\font minus
  \fontdimen4\font\relax}
\providecommand{\BIBforeignlanguage}[2]{{%
\expandafter\ifx\csname l@#1\endcsname\relax
\typeout{** WARNING: IEEEtran.bst: No hyphenation pattern has been}%
\typeout{** loaded for the language `#1'. Using the pattern for}%
\typeout{** the default language instead.}%
\else
\language=\csname l@#1\endcsname
\fi
#2}}
\providecommand{\BIBdecl}{\relax}
\BIBdecl

\bibitem{Foschini1998_OnLimitsOf}
{G. J. Foschini and M. J. Gans}, ``On limits of wireless communications in a
  fading environment when using multiple antennas,'' \emph{Wireless Personal
  Communications}, vol.~6, no.~3, pp. 311--335, Mar 1998.

\bibitem{telatar_99}
{E. Telatar}, ``Capacity of multi-antenna {Gaussian} channels,'' \emph{European
  Transactions on Telecommunications}, vol.~10, no.~6, pp. 585--595, 1999.

\bibitem{Vaze12}
C.~S. {Vaze} and M.~K. {Varanasi}, ``The degree-of-freedom regions of {MIMO}
  broadcast, interference, and cognitive radio channels with no {CSIT},''
  \emph{IEEE Transactions on Information Theory}, vol.~58, no.~8, pp.
  5354--5374, Aug 2012.

\bibitem{Goldsmith2006_OnTheOptimality}
T.~Yoo and A.~Goldsmith, ``On the optimality of multiantenna broadcast
  scheduling using zero-forcing beamforming,'' \emph{IEEE Journal on Selected
  Areas in Communications}, vol.~24, no.~3, pp. 528--541, March 2006.

\bibitem{Weingarten2006_TheCapacityRegion}
H.~Weingarten, Y.~Steinberg, and S.~S. Shamai, ``The capacity region of the
  {G}aussian multiple-input multiple-output broadcast channel,'' \emph{IEEE
  Transactions on Information Theory}, vol.~52, no.~9, pp. 3936--3964, Sept
  2006.

\bibitem{Sharif2007_AComparisionTimeSharingDPC}
M.~Sharif and B.~Hassibi, ``A comparison of time-sharing, {DPC}, and
  beamforming for {MIMO} broadcast channels with many users,'' \emph{IEEE
  Transactions on Communications}, vol.~55, no.~1, pp. 11--15, Jan 2007.

\bibitem{Christensen2008_WeightedSumRate}
{S. S. Christensen and R. Agarwal and E. D. Carvalho and J. M. Cioffi},
  ``Weighted sum-rate maximization using weighted {MMSE} for {MIMO-BC}
  beamforming design,'' \emph{IEEE Transactions on Wireless Communications},
  vol.~7, no.~12, pp. 4792--4799, December 2008.

\bibitem{Zhang2011_AUnifiedTreatmentofSPC}
R.~Zhang and L.~Hanzo, ``A unified treatment of superposition coding aided
  communications: Theory and practice,'' \emph{IEEE Communications Surveys
  Tutorials}, vol.~13, no.~3, pp. 503--520, Third 2011.

\bibitem{Vanka2012_SPCStrategies}
S.~Vanka, S.~Srinivasa, Z.~Gong, P.~Vizi, K.~Stamatiou, and M.~Haenggi,
  ``Superposition coding strategies: Design and experimental evaluation,''
  \emph{IEEE Transactions on Wireless Communications}, vol.~11, no.~7, pp.
  2628--2639, July 2012.

\bibitem{Joudeh2016_SumRateMax}
H.~Joudeh and B.~Clerckx, ``Sum-rate maximization for linearly precoded
  downlink multiuser {MISO} systems with partial {CSIT}: A rate-splitting
  approach,'' \emph{IEEE Transactions on Communications}, vol.~64, no.~11, pp.
  4847--4861, Nov 2016.

\bibitem{ding2014performance}
Z.~Ding, Z.~Yang, P.~Fan, and H.~V. Poor, ``On the performance of
  non-orthogonal multiple access in {5G} systems with randomly deployed
  users,'' \emph{arXiv preprint arXiv:1406.1516}, 2014.

\bibitem{seo2018high}
J.~Seo, Y.~Sung, and H.~Jafarkhani, ``A high-diversity transceiver design for
  {K}-user {MISO} broadcast channels,'' \emph{arXiv preprint arXiv:1807.00114},
  2018.

\bibitem{multicast_1}
\BIBentryALTinterwordspacing
A.~Lee, ``{EMBMS} delivers mobile video to the mass audience,'' 2015. [Online].
  Available:
  \url{https://www.itu.int/en/ITU-D/Regional-Presence/AsiaPacific/Documents/Events/2015/August-MTV/S3B-Allan-Lee.pdf}
\BIBentrySTDinterwordspacing

\bibitem{multicast_2}
V.~K. Thorsten~Lohmar, Michael~Slssingar and S.~Puustinen, ``Delivering content
  with {LTE} broadcast,'' \emph{Ericsson Review}, vol.~1, no.~11, February
  2013.

\bibitem{multicast_3}
3GPP, ``Multimedia broadcast/multicast service ({MBMS}); architecture and
  functional description,'' Tech. Rep. TS 23.246, V8.0.0, June 2007.

\bibitem{Jindal2006_CapacityLimitsofMultipleAnt}
N.~Jindal and Z.~Q. Luo, ``Capacity limits of multiple antenna multicast,'' in
  \emph{2006 IEEE International Symposium on Information Theory}, July 2006,
  pp. 1841--1845.

\bibitem{Abdelkader2010_MultipleAntennaMulticasting}
{A. Abdelkader and A. B. Gershman and N. D. Sidiropoulos}, ``Multiple-antenna
  multicasting using channel orthogonalization and local refinement,''
  \emph{IEEE Transactions on Signal Processing}, vol.~58, no.~7, pp.
  3922--3927, July 2010.

\bibitem{Chiu2009_TransmitPrecoding}
E.~Chiu and V.~Lau, ``Transmit precoding design for multi-antenna multicast
  broadcast services with limited feedback,'' in \emph{2009 IEEE Wireless
  Communications and Networking Conference}, April 2009, pp. 1--6.

\bibitem{Sidiropoulos2006_TransmitBeamforming}
{N. D. Sidiropoulos and T. N. Davidson and Zhi-Quan Luo}, ``Transmit
  beamforming for physical-layer multicasting,'' \emph{IEEE Transactions on
  Signal Processing}, vol.~54, no.~6, pp. 2239--2251, June 2006.

\bibitem{Silva2006_AdaptiveBeamforming}
Y.~C.~B. Silva and A.~Klein, ``Adaptive beamforming and spatial multiplexing of
  unicast and multicast services,'' in \emph{2006 IEEE 17th International
  Symposium on Personal, Indoor and Mobile Radio Communications}, Sept 2006,
  pp. 1--5.

\bibitem{Baek2009_AdaptiveTransmission}
S.~Y. Baek, Y.~J. Hong, and D.~K. Sung, ``Adaptive transmission scheme for
  mixed multicast and unicast traffic in cellular systems,'' \emph{IEEE
  Transactions on Vehicular Technology}, vol.~58, no.~6, pp. 2899--2908, July
  2009.

\bibitem{Yalcin18downlink}
A.~Z. {Yalcin}, M.~{Yuksel}, and I.~{Bahceci}, ``Downlink {MU-MIMO} with {QoS}
  aware transmission: Precoder design and performance analysis,'' \emph{IEEE
  Transactions on Wireless Communications}, vol.~18, no.~2, pp. 969--982, Feb
  2019.

\bibitem{Silva2009_LinearTransmitBeamforming}
Y.~C.~B. Silva and A.~Klein, ``Linear transmit beamforming techniques for the
  multigroup multicast scenario,'' \emph{IEEE Transactions on Vehicular
  Technology}, vol.~58, no.~8, pp. 4353--4367, Oct 2009.

\bibitem{Gao2005_GroupOrientedBeamforming}
Y.~Gao and M.~Schubert, ``Group-oriented beamforming for multi-stream
  multicasting based on quality-of-service requirements,'' in \emph{1st IEEE
  International Workshop on Computational Advances in Multi-Sensor Adaptive
  Processing, 2005.}, Dec 2005, pp. 193--196.

\bibitem{Karidipis2008_QoSMaxMinFairTransmit}
E.~Karipidis, N.~D. Sidiropoulos, and Z.~Q. Luo, ``Quality of service and
  max-min fair transmit beamforming to multiple cochannel multicast groups,''
  \emph{IEEE Transactions on Signal Processing}, vol.~56, no.~3, pp.
  1268--1279, March 2008.

\bibitem{Christopoulos2014_WeightedFair}
D.~Christopoulos, S.~Chatzinotas, and B.~Ottersten, ``Weighted fair multicast
  multigroup beamforming under per-antenna power constraints,'' \emph{IEEE
  Transactions on Signal Processing}, vol.~62, no.~19, pp. 5132--5142, Oct
  2014.

\bibitem{Joudeh2017_RateSplittingforMaxMin}
H.~Joudeh and B.~Clerckx, ``Rate-splitting for max-min fair multigroup
  multicast beamforming in overloaded systems,'' \emph{IEEE Transactions on
  Wireless Communications}, vol.~16, no.~11, pp. 7276--7289, Nov 2017.

\bibitem{Zhou2015_JointMulticast}
H.~Zhou and M.~Tao, ``Joint multicast beamforming and user grouping in massive
  {MIMO} systems,'' in \emph{2015 IEEE International Conference on
  Communications (ICC)}, June 2015, pp. 1770--1775.

\bibitem{Sadeghi2018_MMFMassiveMIMO}
M.~Sadeghi, E.~Björnson, E.~G. Larsson, C.~Yuen, and T.~L. Marzetta, ``Max-min
  fair transmit precoding for multi-group multicasting in massive {MIMO},''
  \emph{IEEE Transactions on Wireless Communications}, vol.~17, no.~2, pp.
  1358--1373, Feb 2018.

\bibitem{christopoulos2015multicast}
D.~Christopoulos, S.~Chatzinotas, and B.~Ottersten, ``Multicast multigroup
  precoding and user scheduling for frame-based satellite communications,''
  \emph{IEEE Transactions on Wireless Communications}, vol.~14, no.~9, pp.
  4695--4707, 2015.

\bibitem{kaliszan2012multigroup}
M.~Kaliszan, E.~Pollakis, and S.~Sta{\'n}czak, ``Multigroup multicast with
  application-layer coding: Beamforming for maximum weighted sum rate,'' in
  \emph{Wireless Communications and Networking Conference (WCNC)}.\hskip 1em
  plus 0.5em minus 0.4em\relax IEEE, 2012, pp. 2270--2275.

\bibitem{ekrem2012outer}
E.~Ekrem and S.~Ulukus, ``An outer bound for the {G}aussian {MIMO} broadcast
  channel with common and private messages,'' \emph{IEEE Transactions on
  Information Theory}, vol.~58, no.~11, pp. 6766--6772, 2012.

\bibitem{Kim2015_DesgnOfUserClustering}
J.~Kim, J.~Koh, J.~Kang, K.~Lee, and J.~Kang, ``Design of user clustering and
  precoding for downlink non-orthogonal multiple access ({NOMA}),'' in
  \emph{MILCOM 2015 - 2015 IEEE Military Communications Conference}, Oct 2015,
  pp. 1170--1175.

\bibitem{senel2017optimal}
K.~Senel and S.~Tekinay, ``Optimal power allocation in {NOMA} systems with
  imperfect channel estimation,'' in \emph{GLOBECOM 2017-2017 IEEE Global
  Communications Conference}.\hskip 1em plus 0.5em minus 0.4em\relax IEEE,
  2017, pp. 1--7.

\bibitem{guo2005mutual}
D.~Guo, S.~Shamai, and S.~Verd{\'u}, ``Mutual information and minimum
  mean-square error in {G}aussian channels,'' \emph{IEEE Transactions on
  Information Theory}, vol.~51, no.~4, pp. 1261--1282, 2005.

\bibitem{matrix_cookbook}
\BIBentryALTinterwordspacing
K.~B. Petersen and M.~S. Pedersen, ``The matrix cookbook,'' nov 2012, version
  20121115. [Online]. Available: \url{http://www2.imm.dtu.dk/pubdb/p.php?3274}
\BIBentrySTDinterwordspacing

\bibitem{Li2003_ExponentialPenaltyMethod}
S.~P. X.S.~Li, ``Solving the finite min-max problem via an exponential penalty
  method,'' \emph{Comput. Tech.}, vol.~8, pp. 3--15, 2003.

\bibitem{CVX}
M.~Grant and S.~Boyd, ``{CVX}: Matlab software for disciplined convex
  programming, version 2.1,'' \url{http://cvxr.com/cvx}, Mar. 2014.

\bibitem{Kaleva2016_DecentralizedSumRateMaximization}
J.~Kaleva, A.~Tölli, and M.~Juntti, ``Decentralized sum rate maximization with
  {QoS} constraints for interfering broadcast channel via successive convex
  approximation,'' \emph{IEEE Transactions on Signal Processing}, vol.~64,
  no.~11, pp. 2788--2802, June 2016.

\bibitem{Joham2005_LinearTransmitProcessing}
M.~Joham, W.~Utschick, and J.~A. Nossek, ``Linear transmit processing in {MIMO}
  communications systems,'' \emph{IEEE Transactions on Signal Processing},
  vol.~53, no.~8, pp. 2700--2712, Aug 2005.

\bibitem{Dao2010_UserSelectionAlgo}
N.~Dao and Y.~Sun, ``User-selection algorithms for multiuser precoding,''
  \emph{IEEE Transactions on Vehicular Technology}, vol.~59, no.~7, pp.
  3617--3622, Sept 2010.

\bibitem{Hunger_Complexity2007}
R.~Hunger, \emph{Floating Point Operations in Matrix-Vector Calculus}.\hskip
  1em plus 0.5em minus 0.4em\relax Munich University of Technology, Inst. for
  Circuit Theory and Signal Processing, 2005.

\bibitem{moon_stirling_math}
T.~Moon and W.~Stirling, \emph{Mathematical Methods and Algorithms for Signal
  Processing}.\hskip 1em plus 0.5em minus 0.4em\relax Prentice Hall, 2000.

\end{thebibliography}
\end{document}